\documentclass[11pt]{article}
\usepackage{amsmath,amsfonts,amssymb,amsthm}
\usepackage{enumerate}
\usepackage{xcolor}
\usepackage{url}
\usepackage{tcolorbox}
\usepackage{hyperref}
\usepackage{multicol, latexsym}
\usepackage{latexsym}
\usepackage{psfrag,import}
\usepackage{verbatim}
\usepackage{color}
\usepackage{epsfig}
\usepackage[outdir=./]{epstopdf}
\usepackage{hyperref}
\hypersetup{
    colorlinks=true,
    linkcolor=blue,
    filecolor=magenta,      
    urlcolor=cyan
}
\usepackage[title]{appendix}
\usepackage{geometry}
\usepackage{mathtools}
\usepackage{enumerate}
\usepackage{enumitem}   
\usepackage{multicol}
\usepackage{booktabs}
\usepackage{enumitem}

\usepackage{parcolumns}
\usepackage{thmtools}
\usepackage{xr}
\usepackage{epstopdf}
\usepackage{mathrsfs}

 \usepackage{subcaption}
 \usepackage{soul}
\usepackage{float}

\usepackage{comment}
\usepackage{authblk}
\usepackage{setspace}

\usepackage{cleveref}

\theoremstyle{plain}
\newtheorem{theorem}{Theorem}[section]
\newtheorem{corollary}[theorem]{Corollary}
\newtheorem{proposition}[theorem]{Proposition}
\newtheorem{lemma}[theorem]{Lemma}

\theoremstyle{definition}
\newtheorem{definition}[theorem]{Definition}
\newtheorem{example}[theorem]{Example}

\newtheorem{remark}[theorem]{Remark}

\theoremstyle{remark}

\numberwithin{equation}{section}

\newcommand\RR{\mathbb{R}}

\newcommand\bla{\boldsymbol{\lambda}}
\newcommand\by{\boldsymbol{y}}

\newcommand\byi{\boldsymbol{y_i}}

\newcommand\byj{\boldsymbol{y_j}}

\newcommand\bep{\boldsymbol{\varepsilon}}

\renewcommand\bf{\boldsymbol{f}}

\newcommand\bk{\boldsymbol{k}}
\newcommand\bw{\boldsymbol{w}}

\newcommand\bg{\boldsymbol{g}}
\newcommand\bx{\boldsymbol{x}}
\newcommand\bv{\boldsymbol{v}}
\newcommand\bz{\boldsymbol{z}}

\newcommand\bA{\boldsymbol{A}}
\newcommand\bB{\boldsymbol{B}}

\newcommand\bq{\boldsymbol{q}}
\newcommand\bp{\boldsymbol{p}}

\newcommand\bJ{\boldsymbol{J}}

\newcommand\hbJ{\hat{\boldsymbol{J}}}

\newcommand{\mK}{\mathcal{K}}
\newcommand{\dK}{\mathcal{K}_{\RR\text{-disg}}}
\newcommand{\pK}{\mathcal{K}_{\text{disg}}}
\newcommand{\mJ}{\mathcal{J}_{\RR}}

\newcommand{\eJ}{\mathcal{J}_{\textbf{0}}}
\newcommand{\mD}{\mathcal{D}_{\textbf{0}}}

\newcommand{\mS}{\mathcal{S}}

\newcommand{\hPsi}{\hat{\Psi}}
\newcommand{\hbx}{\hat{\bx}}
\newcommand{\hbk}{\hat{\bk}}

\newcommand{\hbq}{\hat{\bq}}
\newcommand{\hmJ}{\hat{\mJ}}

\newcommand{\defi}{\textbf}
\DeclareMathOperator{\spn}{span}

\begin{document}

\title{ 
The Dimension of the Disguised Toric Locus of a Reaction Network
}

\author[1]{
         Gheorghe Craciun
}
\author[2]{
        Abhishek Deshpande
}
\author[3]{
        Jiaxin Jin
}
\affil[1]{\small Department of Mathematics and Department of Biomolecular Chemistry, University of Wisconsin-Madison}
\affil[2]{Center for Computational Natural Sciences and Bioinformatics, \protect \\
 International Institute of Information Technology Hyderabad}
\affil[3]{\small Department of Mathematics, University of Louisiana at Lafayette}

\date{} 

\maketitle

\begin{abstract}
Under mass-action kinetics, complex-balanced systems emerge from biochemical reaction networks and exhibit stable and predictable dynamics.
For a reaction network $G$, the associated dynamical system is called \emph{disguised toric} if it can yield a complex-balanced realization on a possibly different network $G_1$.
This concept extends the robust properties of toric systems to those that are not inherently toric.
In this work, we study the \emph{disguised toric locus} of a reaction network — i.e., the set of positive rate constants that make the corresponding mass-action system disguised toric. Our primary focus is to compute the exact dimension of this locus.
We subsequently apply our results to Thomas-type and circadian clock models.
\end{abstract}

\begin{NoHyper}
\tableofcontents
\end{NoHyper}

\section{Introduction}

Mathematical models of biochemical interaction networks can generally be described by {\em polynomial dynamical systems}. These dynamical systems are ubiquitous in models of biochemical reaction networks, infectious diseases, and population dynamics~\cite{craciun2022homeostasis,deshpande2014autocatalysis}.
However, analyzing these systems is a challenging problem in general. Classical nonlinear dynamical properties like multistability, oscillations, or chaotic dynamics are difficult to examine~\cite{Ilyashenko2002, yu2018mathematical}.

Studying the dynamical properties of reaction networks is crucial for understanding the behavior of chemical and biological systems. In this paper, we will focus on a class of dynamical systems generated by reaction networks called {\em complex-balanced systems} (also known as {\em toric dynamical systems}~\cite{CraciunDickensteinShiuSturmfels2009} owing to their connection with toric varieties~\cite{dickenstein2020algebraic}). Complex-balanced systems are known to exhibit remarkably robust dynamics, which {\em rules out} multistability, oscillations, and even chaotic dynamics~\cite{horn1972general}. 
More specifically, there exists a strictly convex Lyapunov function, 
which implies that all positive steady states are locally asymptotically stable~\cite{horn1972general, yu2018mathematical}. In addition, a unique positive steady state exists within each affine invariant polyhedron. They are also related to the \emph{Global Attractor Conjecture}~\cite{CraciunDickensteinShiuSturmfels2009} which states that complex-balanced dynamical systems have a globally attracting steady state within each stoichiometric compatibility class. Several special cases of this conjecture have been proved~\cite{anderson2011proof,gopalkrishnan2014geometric, pantea2012persistence, craciun2013persistence, boros2020permanence}, and a proof in full generality has been proposed in~\cite{craciun2015toric}.

An important limitation of the classical theory of complex-balanced systems is that to be applicable for a large set of parameter values (i.e., choices of reaction rate constants) the reaction network under consideration must satisfy two special properties: {\em weak reversibility} and {\em low deficiency} (see \cite{yu2018mathematical} for definitions). Our focus here will be on understanding how one can take advantage of the notion of {\em dynamical equivalence} in order to greatly relax both of these restrictions. 

Dynamical equivalence relies on the fact that two different reaction networks can generate the same dynamics for well-chosen parameter values. This phenomenon has also been called \emph{macro-equivalence}~\cite{horn1972general} or {\em confoundability}~\cite{craciun2008identifiability}. 
Recently, this phenomenon has found applications in the design of efficient algorithms for finding weakly reversible single linkage class and weakly reversible deficiency one realizations~\cite{WR_df_1, WR_DEF_THM}. Moreover, it has also been used to show the existence of infinitely many positive states for weakly reversible and endotactic dynamical systems~\cite{boros2020weakly,kothari2024endotactic}. More recently, it has been used to generate the necessary and sufficient conditions for the existence of realizations using weakly reversible dynamical systems~\cite{kothari2024realizations}.

In this paper, we consider the notion of a disguised toric locus for a given reaction network $G$.  
The \emph{disguised toric locus} is the set of positive reaction rate vectors in $G$ for which the corresponding dynamical system can be realized as a complex-balanced system by a network $G_1$.
In other words, this locus consists of positive reaction rate vectors $\bk$ such that the mass-action system $(G, \bk)$ is dynamically equivalent to a complex-balanced system $(G_1, \bk_1)$.
Additionally, if the rate constants are allowed to take any real values, we refer to the set of reaction rate vectors in $G$ that satisfy this property as the \emph{$\mathbb{R}$-disguised toric locus} of $G$.

The concept of a disguised toric locus was first introduced in \cite{2022disguised}. Since then, several general properties of both the disguised toric locus and the $\mathbb{R}$-disguised toric locus have been established. For example, it was demonstrated in \cite{haque2022disguised} that the disguised toric locus is invariant under invertible affine transformations of the network. 
Furthermore, \cite{disg_1} showed that both loci are path-connected, and \cite{disg_2} provided a lower bound on the dimension of the $\mathbb{R}$-disguised toric locus.

Consider for example the Thomas-type model (E-graph $G$) shown in Figure \ref{fig:thomas_model_intro}.

\begin{figure}[!ht]
\centering
\includegraphics[scale=0.7]{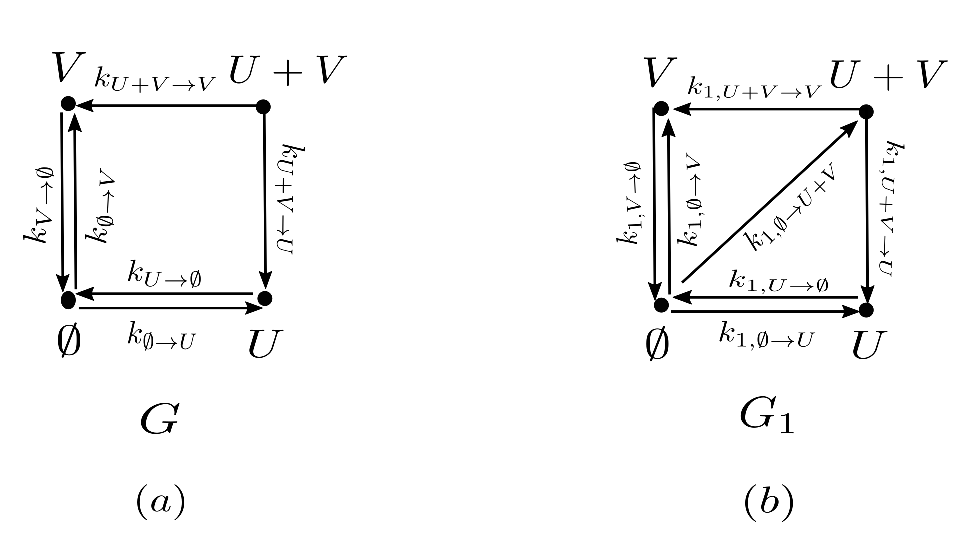}
\caption{ 
(a) The E-graph $G$ represents a Thomas-type model, with all edges labeled by the reaction rate constants $\bk$.
(b) The E-graph $G_1$ is weakly reversible, with all edges labeled by the reaction rate constants $\bk_1$.
The mass-action system $(G_1, \bk_1)$ is complex-balanced.
}
\label{fig:thomas_model_intro}
\end{figure}

Since $G$ is not weakly reversible, the system $(G, \bk)$ is not complex-balanced, so classical complex-balanced theory offers limited insight into the dynamics of $(G, \bk)$.
However, by direct computation, $(G, \bk)$ is dynamically equivalent to the complex-balanced system $(G_1, \bk_1)$, which enables us to deduce its dynamical properties.
Thus, $\bk$ can be viewed as a “good” reaction rate vector for $G$. The disguised toric locus of $G$ consists of such reaction rate vectors $\bk$.

In this paper, we develop a general framework to compute the exact dimensions of both the disguised toric locus and the $\mathbb{R}$-disguised toric locus of a reaction network. Building on \cite{disg_2}, we construct a mapping on the $\mathbb{R}$-disguised toric locus of $G$ and show that this mapping is a homeomorphism, allowing us to determine the dimensions of both the disguised toric locus and the $\mathbb{R}$-disguised toric locus.
When applied to Figure \ref{fig:thomas_model_intro}, the disguised toric locus of $G$ is shown to be full-dimensional, significantly larger than its toric locus, which is empty (see details in Example \ref{ex:thomas}).

\bigskip

\textbf{Structure of the paper.}
In Section~\ref{sec:reaction_networks}, we introduce the basic terminology of reaction networks.
Section~\ref{sec:flux_systems} presents flux systems and analyzes their properties.
In Section~\ref{sec:disguised_locus}, we recall the key concepts of the toric locus, the $\RR$-disguised toric locus, and the disguised toric locus. 
Section~\ref{sec:map} constructs a continuous bijective map $\hPsi$ connecting the $\RR$-disguised toric locus to a specific flux system.
In Section~\ref{sec:continuity}, we first establish key lemmas \ref{lem:key_1} - \ref{lem:key_4} and then use them to prove that $\hPsi$ is a homeomorphism in Theorem \ref{thm:hpsi_homeo}.
Section~\ref{sec:dimension} leverages this homeomorphism to establish precise bounds on the dimension of the disguised toric locus and the $\RR$-disguised toric locus, as shown in Theorem~\ref{thm:dim_kisg_main}.
In Section~\ref{sec:applications}, we apply our results to Thomas-type models and circadian clock models, showing both disguised toric loci are full-dimensional.
Finally, Section~\ref{sec:discussion} summarizes our findings and outlines potential directions for future research.

\bigskip

\textbf{Notation.}
We let $\mathbb{R}_{\geq 0}^n$ and $\mathbb{R}_{>0}^n$ denote the set of vectors in $\mathbb{R}^n$ with non-negative entries and positive entries respectively. For vectors $\bx = (\bx_1, \ldots, \bx_n)^{\intercal}\in \RR^n_{>0}$ and $\by = (\by_1, \ldots, \by_n)^{\intercal} \in \RR^n$, we define:
\begin{equation} \notag
\bx^{\by} = \bx_1^{y_{1}} \ldots \bx_n^{y_{n}}.
\end{equation}
For any two vectors $\bx, \by \in \RR^n$, we define $\langle \bx, \by \rangle = \sum\limits^{n}_{i=1} x_i y_i$.
For E-graphs (see Definition \ref{def:e-graph}), we always let $G, G'$ denote arbitrary E-graphs, and let $G_1$ denote a weakly reversible E-graph.

\section{Reaction networks}
\label{sec:reaction_networks}

We start with the introduction of the concept of a {\em reaction network} as a directed graph in Euclidean space called  {\em E-graph}, and describe some of its properties.

\begin{definition}[\cite{craciun2015toric, craciun2019polynomial,craciun2020endotactic}]
\label{def:e-graph}
\begin{enumerate}[label=(\alph*)]
\item A \textbf{reaction network} $G=(V,E)$ is a directed graph, also called a \textbf{Euclidean embedded graph} (or
\textbf{E-graph}), such that $V \subset \mathbb{R}^n$ is a finite set of \textbf{vertices} and the set $E\subseteq V\times V$ represents a finite set of \textbf{edges}. We assume that there are neither self-loops nor isolated vertices in $G=(V, E)$. 

\item A directed edge $(\by,\by')\in E$ connecting two vertices $\by, \by' \in V$ is denoted by $\by \rightarrow \by' \in E$ and represents a reaction in the network. 
Here $\by$ is called the \textbf{source vertex}, and $\by'$ is called the \textbf{target vertex}.
Further, the difference vector $\by' - \by \in\mathbb{R}^n$ is called the \textbf{reaction vector}. 
\end{enumerate}
\end{definition}

\begin{definition}
Consider an E-graph $G=(V,E)$. 	Then 
\begin{enumerate}[label=(\alph*)]
\item $G$ is \textbf{weakly reversible}, if every reaction in $G$ is part of an oriented cycle. 

\item $G$ is a \textbf{(directed) complete} graph, if $\by\rightarrow \by'\in E$ for every two distinct vertices $\by, \by'\in V$.

\item An E -graph $G' = (V', E')$ is a \textbf{subgraph} of $G$ (denoted by $G' \subseteq G$), if  $V' \subseteq V$ and $E' \subseteq E$. 
In addition, we let $G' \sqsubseteq G$ denote that $G'$ is a weakly reversible subgraph of $G$.

\item We denote the \defi{complete graph on $G$} by $G_c$, which is obtained by connecting every pair of source vertices in $V$. One can check that $G_c$ is weakly reversible and $G \subseteq G_c$.
\end{enumerate}
\end{definition}

\begin{figure}[!ht]
\centering
\includegraphics[scale=0.4]{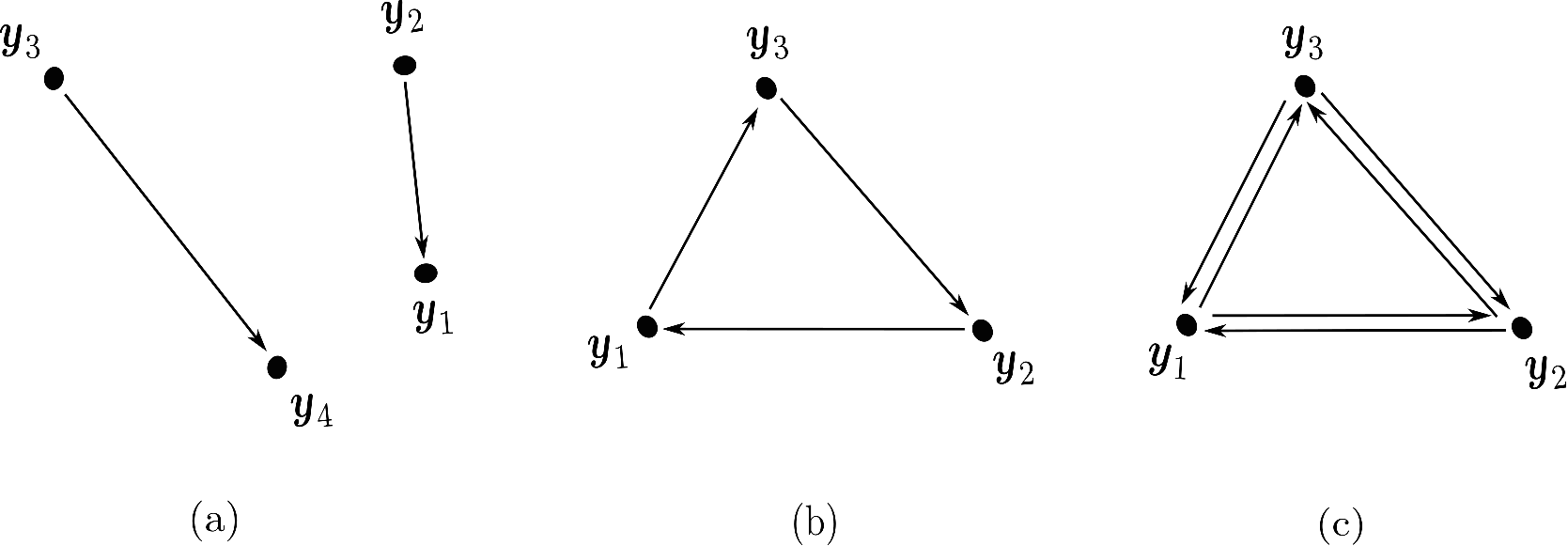}
\caption{\small (a) An E-graph with two reactions. The stoichiometric subspace corresponding to this graph is $\RR^2$. (b) A weakly reversible E-graph. (c) A directed complete E-graph with three vertices. Note that the E-graph in (b) is a weakly reversible subgraph of the E-graph in (c).}
\label{fig:e-graph}
\end{figure} 

\begin{definition}[\cite{adleman2014mathematics,guldberg1864studies,voit2015150,gunawardena2003chemical,yu2018mathematical,feinberg1979lectures}]
Consider an E-graph $G=(V,E)$. Let $k_{\by\to \by'}$ denote the \textbf{reaction rate constant} corresponding to the reaction $\by\to \by'\in E$.
Further, we let ${\bk} :=(k_{\by\to \by'})_{\by\to \by' \in E} \in \mathbb{R}_{>0}^{E}$ denote the \textbf{vector of reaction rate constants} (\textbf{reaction rate vector}).
The \textbf{associated mass-action system} generated by $(G, \bk)$ 
on $\RR^n_{>0}$ is given by
\begin{equation}
\label{def:mas_ds}
\frac{d\bx}{dt} = \displaystyle\sum_{\by \rightarrow \by' \in E}k_{\by\rightarrow\by'}{\bx}^{\by}(\by'-\by).
\end{equation}
We denote the \defi{stoichiometric subspace} of $G$ by $\mathcal{S}_G$, which is
\begin{equation} \notag
\mathcal{S}_G = \spn \{ \by' - \by: \by \rightarrow \by' \in E \}.
\end{equation}
\cite{sontag2001structure} shows that if $V \subset \mathbb{Z}_{\geq 0}^n$, the positive orthant $\mathbb{R}_{>0}^n$ is forward-invariant under system \eqref{def:mas_ds}. 
Any solution to \eqref{def:mas_ds} with initial condition $\bx_0 \in \mathbb{R}_{>0}^n$ and $V \subset \mathbb{Z}_{\geq 0}^n$, is confined to $(\bx_0 + \mathcal{S}_G) \cap \mathbb{R}_{>0}^n$. Thus, the set $(\bx_0 + \mathcal{S}_G) \cap \mathbb{R}_{>0}^n$ is called the \textbf{invariant polyhedron} of $\bx_0$.
\end{definition}

\begin{definition}
Let $(G, \bk)$ be a mass-action system. 

\begin{enumerate}[label=(\alph*)]
\item A point $\bx^* \in \mathbb{R}^n_{>0}$ is called a \defi{positive steady state} of the system if 
\begin{equation}
\label{eq:steady_statez} 
\displaystyle\sum_{\by\rightarrow \by' \in E } k_{\by\rightarrow\by'}{(\bx^*)}^{\by}(\by'-\by)=0.
\end{equation}

\item A point $\bx^* \in \mathbb{R}^n_{>0}$ is called a \defi{complex-balanced steady state} of the system if for every vertex $\by_0 \in V$,
\begin{eqnarray}  \notag
\sum_{\by_0 \rightarrow \by \in E} k_{\by_0 \rightarrow \by} {(\bx^*)}^{\by_0}
= \sum_{\by' \rightarrow \by_0 \in E} k_{\by' \rightarrow \by_0} {(\bx^*)}^{\by'}.
\end{eqnarray}
Further, if the mass-action system $(G, \bk)$ admits a complex-balanced steady state, then it is called a \defi{complex-balanced  (dynamical) system} or \defi{toric dynamical system}. 
\end{enumerate}
\end{definition}

\begin{remark}
\label{rmk:complex_balance_property}

Complex-balanced systems are known to exhibit robust dynamical properties.
As mentioned in the introduction, they are connected to the \emph{Global Attractor Conjecture}, which proposes that complex-balanced systems possess a globally attracting steady state within each stoichiometric compatibility class. Several important special cases of this conjecture and related open problems have been proven. In particular, it has been shown that complex-balanced systems consisting of a single linkage class admit a globally attracting steady state \cite{anderson2011proof}.
Additionally, two- and three-dimensional endotactic networks are known to be permanent \cite{craciun2013persistence}. Strongly endotactic networks have also been proven to be permanent \cite{gopalkrishnan2014geometric}. 
Furthermore, complex-balanced systems that are permanent always admit a globally attracting steady state \cite{yu2018mathematical}.
\end{remark}

\begin{theorem}[\cite{horn1972general}]
\label{thm:cb}

Consider a complex-balanced system $(G, \bk)$. Then 
\begin{enumerate}
\item[(a)] The E-graph $G = (V, E)$ is weakly reversible.

\item[(b)] Every positive steady state is a complex-balanced steady state.
Given any $\bx_0 \in \mathbb{R}_{>0}^n$,
there is exactly one steady state within 
the invariant polyhedron $(\bx_0 + \mathcal{S}_G) \cap \mathbb{R}_{>0}^n$. 
\end{enumerate}
\end{theorem}

\begin{theorem}[\cite{johnston2012topics}]
\label{thm:jacobian}

Consider a weakly reversible E-graph $G = (V, E)$ with the stoichiometric subspace $\mS_G$.
Suppose $(G, \bk)$ is a complex-balanced system given by
\begin{equation}
\label{eq:jacobian}
\frac{\mathrm{d} \bx}{\mathrm{d} t} 
= \bf (\bx) = \displaystyle\sum_{\by\rightarrow \by' \in E} k_{\by\rightarrow\by'}{\bx}^{\by}(\by'-\by).
\end{equation}
For any steady state $\bx^* \in \RR^n_{>0}$ of the system \eqref{eq:jacobian}, then
\begin{equation}
\label{eq:jacobian_ker}
\Big( \ker \big( \mathbf{J}_{\bf} |_{\bx = \bx^*} \big) \Big)^{\perp} = \mS_G,
\end{equation}
where $\mathbf{J}_{\bf}$ represents the Jacobian matrix of $\bf (\bx)$.
\end{theorem}

\begin{definition}
\label{def:de}
Consider two mass-action systems $(G,\bk)$ and $(G',\bk')$. Then $(G,\bk)$ and $(G',\bk')$ are said to be \defi{dynamically equivalent} if for every vertex\footnote{ Note that when $\by_0 \not\in V$ or $\by_0 \not\in V'$, the corresponding side is considered as an empty sum} $\by_0 \in V \cup V'$,
\begin{eqnarray} \notag
\displaystyle\sum_{\by_0 \rightarrow \by\in E} k_{\by_0 \rightarrow \by} (\by - \by_0) 
= \displaystyle\sum_{\by_0 \rightarrow \by'\in E'} k'_{\by_0 \rightarrow\by'} (\by' - \by_0).
\end{eqnarray}
We let $(G,\bk)\sim (G', \bk')$ denote that two mass-action systems $(G,\bk)$ and $(G',\bk')$ are dynamically equivalent.
\end{definition}

\begin{remark}[\cite{horn1972general,craciun2008identifiability,deshpande2022source}]
\label{rmk:de_ss}
Following Definition \ref{def:de}, two mass-action systems $(G, \bk)$ and $(G', \bk')$ are dynamically equivalent if and only if for all $\bx \in \RR_{>0}^{n}$,
\begin{equation}
\label{eq:eqDE}
\sum_{\by_i \to \by_j \in E} k_{\by_i  \to \by_j} \bx^{\by_i} (\by_j - \by_i) 
= \sum_{\by'_i \to \by'_j \in E'} k'_{\by'_i  \to \by'_j} \bx^{\by'_i} (\by'_j - \by'_i),
\end{equation}
and thus two dynamically equivalent systems share the same set of steady states.
\end{remark}

\begin{definition} 
\label{def:d0}
Consider an E-graph $G=(V, E)$. Let $\bla = (\lambda_{\by \to \by'})_{\by \to \by' \in E} \in \RR^{|E|}$. The set $\mD(G)$ is defined as
\begin{equation} \notag
\mD (G):=
\{\bla  \in \RR^{|E|} \, \Big| \, \sum_{\by_0 \to \by \in E} \lambda_{\by_0  \to \by} (\by - \by_0) = \mathbf{0} \ \text{for every vertex } \by_0 \in V
\}.
\end{equation}
We can check that $\mD (G)$ is a linear subspace of $\RR^E$.
\end{definition} 

\begin{lemma}[\cite{disg_2}]
\label{lem:d0}

Consider two mass-action systems $(G, \bk)$ and $(G, \bk')$. 
Then $\bk' - \bk \in \mD (G)$ if and only if $(G, \bk) \sim (G, \bk')$.    
\end{lemma}

\section{Flux systems} \label{sec:flux_systems}

Due to the non-linearity of the dynamical systems, we now introduce linear systems arising from the network structure: the flux systems, and the complex-balanced flux systems, and study their properties.

\begin{definition}
Consider an E-graph $G=(V, E)$. Then 
\begin{enumerate}[label=(\alph*)]
\item Let $J_{\by \to \by'} > 0$ denote the \textbf{flux} corresponding to the edge $\by \to \by'\in E$. 
Further, we let $\bJ = (J_{\by \to \by'})_{\by \to \by' \in E} \in \RR_{>0}^E$ denote the \textbf{flux vector} corresponding to the E-graph $G$. 
The \defi{associated flux system} generated by $(G, \bJ)$ is given by
\begin{equation} \label{eq:flux}
 \frac{\mathrm{d} \bx}{\mathrm{d} t} 
= \sum_{\byi \to \byj \in E} J_{\byi \to \byj} 
(\byj - \byi).
\end{equation}

\item Consider two flux systems $(G,\bJ)$ and $(G', \bJ')$. Then $(G,\bJ)$ and $(G', \bJ')$ are said to be \defi{flux equivalent} if for every vertex\footnote{Note that when $\by_0 \not\in V$ or $\by_0 \not\in V'$, the corresponding side is considered as an empty sum} $\by_0 \in V \cup V'$,
\begin{equation} \notag
\sum_{\by_0 \to \by \in E} J_{\by_0 \to \by} (\by - \by_0) 
= \sum_{\by_0 \to \by' \in E'} J'_{\by_0 \to \by'} (\by' - \by_0).
\end{equation}
We let $(G, \bJ) \sim (G', \bJ')$ denote that two flux systems $(G, \bJ)$ and $(G', \bJ')$ are flux equivalent. 
\end{enumerate}
\end{definition}

\begin{definition}
Let $(G,\bJ)$ be a flux system. A flux vector $\bJ \in \RR_{>0}^E$ is called a \defi{steady flux vector} to $G$ if 
\begin{equation} \notag
\frac{\mathrm{d} \bx}{\mathrm{d} t} 
= \sum_{\byi \to \byj \in E} J_{\byi \to \byj} 
(\byj - \byi) = \mathbf{0}.
\end{equation}
A steady flux vector $\bJ\in \RR^{E}_{>0}$ is called a \defi{complex-balanced flux vector} to $G$ if for every vertex $\by_0 \in V$, 
\begin{eqnarray} \notag
\sum_{ \by_0 \to \by \in E} J_{\by_0 \to \by} 
= \sum_{\by' \to \by_0 \in E} J_{\by' \to \by_0},
\end{eqnarray} 
and then $(G, \bJ)$ is called a \defi{complex-balanced flux system}. 
Further, let $\mathcal{J}(G)$ denote the set of all complex-balanced flux vectors to $G$ as follows:
\begin{equation} \notag
\mathcal{J}(G):=
\{\bJ \in \RR_{>0}^{E} \mid \bJ  \text{ is a complex-balanced flux vector to $G$} \}.
\end{equation}
\end{definition}

\begin{definition} 
\label{def:j0}
Consider an E-graph $G=(V, E)$. Let $\bJ = ({J}_{\byi \to \byj})_{\byi \to \byj \in E} \in \RR^E$.
The set $\eJ (G)$ is defined as
\begin{equation} \label{eq:J_0}
\eJ (G): =
\{{\bJ} \in \mD (G) \, \bigg| \, \sum_{\by \to \by_0 \in E} {J}_{\by \to \by_0} 
= \sum_{\by_0 \to \by' \in E} {J}_{\by_0 \to \by'} \ \text{for every vertex } \by_0 \in V
\}.
\end{equation}
Note that $\eJ(G) \subset \mD (G)$ is a linear subspace of $\RR^E$.
\end{definition} 

\begin{lemma}[\cite{disg_2}]
\label{lem:j0}
Let $(G, \bJ)$ and $(G, \bJ')$ be two flux systems. Then 
\begin{enumerate}
\item[(a)] $(G, \bJ) \sim (G, \bJ')$ if and only if $\bJ' - \bJ \in \mD (G)$.

\item[(b)] If $(G, \bJ)$ and $(G, \bJ')$ are both complex-balanced flux systems, then $(G, \bJ) \sim (G, \bJ')$ if and only if $\bJ' - \bJ \in \eJ(G)$.
\end{enumerate} 
\end{lemma}

\begin{proposition}[\cite{craciun2020efficient}]
\label{prop:craciun2020efficient}
Consider two mass-action systems $(G, \bk)$ and $(G', \bk')$. Let $\bx \in \RR_{>0}^n$. 
Define the flux vector $\bJ (\bx) = (J_{\by \to \by'})_{\by \to \by' \in E}$ on $G$, such that for every $\by \to \by' \in E$,
\begin{equation} \notag
J_{\by \to \by'} = k_{\by \to \by'} \bx^{\by}.
\end{equation}
Further, define the flux vector $\bJ' (\bx) = (J'_{\by \to \by'})_{\by \to \by' \in E'}$ on $G'$, such that for every $\by \to \by' \in E$,
\begin{equation} \notag
J'_{\by \to \by'} = k'_{\by \to \by'} \bx^{\by}.
\end{equation} 
Then the following are equivalent:
\begin{enumerate}
\item[(a)] The mass-action systems $(G, \bk)$ and $(G', \bk')$ are dynamically equivalent.

\item[(b)] The flux systems $(G, \bJ(\bx))$ and $(G', \bJ')$ are flux equivalent for all $\bx \in \RR_{>0}^n$.

\item[(c)] The flux systems $(G, \bJ(\bx))$ and $(G', \bJ'(\bx))$ are flux equivalent for some $\bx \in \RR_{>0}^n$
\end{enumerate}
\end{proposition}

\section{Toric locus, disguised toric locus and \texorpdfstring{$\RR$}{R}-disguised toric locus} 
\label{sec:disguised_locus}

In this section, we introduce the key concepts in this paper: the Toric locus, the Disguised toric locus, and the $\RR$-disguised toric locus. 

\begin{definition}[\cite{disg_2}]
\label{def:mas_realizable}
Let $G=(V, E)$ be an E-graph. Consider a dynamical system 
\begin{equation} \label{eq:realization_ode}
\frac{\mathrm{d} \bx}{\mathrm{d} t} 
= \bf (\bx).
\end{equation}
It is said to be \defi{$\RR$-realizable} (or has a \defi{$\RR$-realization}) on $G$, if there exists some $\bk \in \mathbb{R}^{E}$ such that
\begin{equation} \label{eq:realization}
\bf (\bx) =
\sum_{\by_i \rightarrow \by_j \in E}k_{\by_i \rightarrow \by_j} \bx^{\by_i}(\by_j - \by_i).
\end{equation}
Further, if $\bk \in \mathbb{R}^{E}_{>0}$ in \eqref{eq:realization}, the system \eqref{eq:realization_ode} is said to be \defi{realizable} (or has a \defi{realization}) on $G$.
\end{definition}

\begin{definition}
Consider an E-graph $G=(V, E)$.
\begin{enumerate}
\item[(a)] Define the \defi{toric locus} of $G$ as
\begin{equation} \notag
\mK (G) := \{ \bk \in \mathbb{R}_{>0}^{E} \ \big| \ \text{the mass-action system generated by } (G, \bk) \ \text{is toric} \}.
\end{equation}

\item[(b)] Consider a dynamical system 
\begin{equation} \label{eq:def_cb_realization}
 \frac{\mathrm{d} \bx}{\mathrm{d} t} 
= \bf (\bx).
\end{equation}
It is said to be \defi{disguised toric} on $G$ if it is realizable on $G$ for some $\bk \in \mK (G)$. Further, we say the system 
\eqref{eq:def_cb_realization} has a \defi{complex-balanced realization} on $G$.
\end{enumerate}
\end{definition}

\begin{definition}
\label{def:de_realizable}
Consider two E-graphs $G =(V,E)$ and $G' =(V', E')$. 
\begin{enumerate}
\item[(a)] Define the set $\mK_{\RR}(G', G)$ as 
\begin{equation} \notag
\mK_{\RR}(G', G) := \{ \bk' \in \mK (G') \ \big| \ \text{the mass-action system } (G', \bk' ) \ \text{is $\RR$-realizable on } G \}.
\end{equation}

\item[(b)] Define the set $\dK(G, G')$ as
\begin{equation} \notag
\dK(G, G') := \{ \bk \in \mathbb{R}^{E} \ \big| \ \text{the dynamical system} \ (G, \bk) \ \text{is disguised toric on } G' \}.
\end{equation} 
Note that $\bk$ may have negative or zero components.

\item[(c)] Define the \defi{$\RR$-disguised toric locus} of $G$ as
\begin{equation} \notag
\dK(G) := \displaystyle\bigcup_{G' \sqsubseteq G_{c}} \ \dK(G, G').
\end{equation}
Note that in the above definition of $\RR$-disguised toric locus of $G$, we take a union over only those E-graphs which are weakly reversible subgraphs of $G_c$. This follows from a result in~\cite{craciun2020efficient} which asserts that if a dynamical system generated by $G$ has a complex-balanced realization using some graph $G_1$, then it also has a complex-balanced realization using $G'\sqsubseteq G_{c}$.

\item[(d)]
Define the set $\pK (G, G')$ as
\begin{equation} \notag
\pK (G, G') := \dK(G, G') \cap \mathbb{R}^{E}_{>0}.
\end{equation} 
Further, define the \defi{disguised toric locus} of $G$ as
\begin{equation} \notag
\pK (G) := \displaystyle\bigcup_{G' \sqsubseteq G_{c}} \ \pK(G, G').
\end{equation}
Similar to the $\RR$-disguised toric locus, it is sufficient for us to include those E-graphs which are weakly reversible subgraphs of $G_c$~\cite{craciun2020efficient}.

\end{enumerate}

\end{definition}

\begin{lemma}[\cite{disg_2}]
\label{lem:semi_algebaic}
Let $G = (V, E)$ be an E-graph. \begin{enumerate}
\item[(a)] Suppose that $G_1 = (V_1, E_1)$ is a weakly reversible E-graph, then $\dK(G,G_1)$ and $\pK(G,G_1)$ are semialgebraic sets.

\item[(b)] Both $\dK(G)$ and $\pK(G)$ are semialgebraic sets.
\end{enumerate}
\end{lemma} 

\begin{proof}

For part $(a)$, following Lemma 3.6 in \cite{disg_2}, we obtain that $\dK(G, G_1)$ is a semialgebraic set. The positive orthant is also a semialgebraic set since it can be defined by polynomial inequalities on all components. Since finite intersections of semialgebraic sets are semialgebraic sets, together with Definition \ref{def:de_realizable}, we conclude that $\pK(G, G_1)$ is a semialgebraic set.

\smallskip

For part $(b)$, since finite unions of semialgebraic sets are semialgebraic sets~\cite{coste2000introduction}, together with Definition \ref{def:de_realizable} and part $(a)$, we conclude that $\dK(G)$ and $\pK(G)$ are semialgebraic sets.
\end{proof}

\begin{remark}[\cite{lee2010introduction}]
\label{rmk:semi_algebaic}
From Lemma \ref{lem:semi_algebaic} and \cite{lee2010introduction}, on a dense open subset of any semialgebraic set $\dK(G, G_1)$ or $\pK(G, G_1)$, it is locally a \textbf{submanifold}. 
The dimension of $\dK(G, G_1)$ or $\pK(G, G_1)$ can be defined to be the largest dimension at points at which it is a submanifold.
\end{remark}

\begin{remark}
\label{rmk:mJ_dK}
Let $G_1 = (V_1, E_1)$ be a weakly reversible E-graph and let $G = (V, E)$ be an E-graph. From Definition \ref{def:de_realizable}, it follows that $\dK (G, G_1)$ is empty if and only if $\mK_{\RR} (G_1, G)$ is empty. 
\end{remark}

Analogous to the $\RR$-disguised toric locus, we also introduce the $\RR$-realizable complex-balanced flux system, which plays a crucial role in the rest of the paper.

\begin{definition}
\label{def:flux_realizable}
Consider a flux system $(G', \bJ')$. It is said to be \defi{$\RR$-realizable} on $G$ if there exists some $\bJ \in \mathbb{R}^{E}$, such that for every vertex\footnote{Note that when $\by_0 \not\in V$ or $\by_0 \not\in V'$, the corresponding side is considered as an empty sum} $\by_0 \in V \cup V'$,
\begin{equation} \notag
\sum_{\by_0 \to \by \in E} J_{\by_0 \to \by} 
(\by - \by_0) 
= \sum_{\by_0 \to \by' \in E'} J'_{\by_0 \to \by'} 
(\by' - \by_0).
\end{equation}
Further, define the set $\mJ (G', G)$ as
\begin{equation} \notag
\mJ (G', G) := \{ \bJ' \in \mathcal{J} (G') \ \big| \ \text{the flux system } (G', \bJ') \ \text{is $\RR$-realizable on } G \}.
\end{equation}
Proposition \ref{prop:craciun2020efficient} implies that $\dK (G, G')$ is empty if and only if $\mJ(G', G)$ is empty.
\end{definition} 

\begin{lemma}[{\cite[Lemma 2.33]{disg_2}}]
\label{lem:j_g1_g_cone}
Consider a weakly reversible E-graph $G_1 = (V_1, E_1)$ and let $G = (V, E)$ be an E-graph. Then we have the following:
\begin{enumerate}
\item[(a)] There exists a vectors $\{ \bv_1, \bv_2, \ldots, \bv_k \} \subset \RR^{|E_1|}$, such that \begin{equation} \label{j_g1_g_generator}
\mJ (G_1, G) = \{ a_1 \bv_1 + \cdots a_k \bv_k \ | \ a_i \in \RR_{>0}, \bv_i \in \RR^{|E_1|} \}.
\end{equation} 

\item[(b)] $\dim (\mJ (G_1, G)) = \dim ( \spn \{ \bv_1, \bv_2, \ldots, \bv_k \} )$.

\item[(c)] If $\mJ (G_1, G) \neq \emptyset$, then
\[
\eJ(G_1) \subseteq \spn \{ \bv_1, \bv_2, \ldots, \bv_k \}.
\]
\end{enumerate}
\end{lemma}

\section{The map \texorpdfstring{$\hPsi$}{hPsi}}
\label{sec:map}

The goal of this section is to study the properties of a map $\hat{\Psi}$ (see Definition \ref{def:hpsi}) that relates the sets $\dK(G, G_1)$ and $\hat{\mJ} (G_1, G)$ (see Equation \eqref{def:hat_j_g1_g}). In particular, we show the map $\hat{\Psi}$ is bijective and continuous.

\paragraph{Notation.}
We introduce the following notation that will be used for the entire section. Let $G = (V, E)$ be an E-graph. Let $b$ denote the dimension of the linear subspace $\mD(G)$, and denote an orthonormal basis of $\mD(G)$ by
\[
\{\bB_1, \bB_2, \ldots, \bB_b\}.
\]
Further, we consider $G_1 = (V_1, E_1)$ to be a weakly reversible E-graph.
Let $a$ denote the dimension of the linear subspace $\eJ(G_1)$, and denote an orthonormal basis of $\eJ(G_1)$ by 
\[
\{\bA_1, \bA_2, \ldots, \bA_a \}.
\]
\qed

\medskip

Recall the set $\mJ (G_1,G)$. Now we define the set $\hat{\mJ} (G_1,G) \subset \RR^{|E_1|}$ as 
\begin{equation}
\label{def:hat_j_g1_g}
\hat{\mJ} (G_1,G) = \{ \bJ + \sum\limits^a_{i=1} w_i \bA_i \ | \ \bJ \in \mJ (G_1,G), \text{ and } w_i \in \RR \text{ for } 1 \leq i \leq a \}.
\end{equation}
Further, we define the set $\hat{\mathcal{J}} (G_1) \subset \RR^{|E_1|}$ as 
\begin{equation} 
\label{def:hat_j_g1}
\hat{\mathcal{J}} (G_1) = \{\bJ \in \RR^{E} \mid \sum_{\by \to \by_0 \in E} J_{\by \to \by_0} 
= \sum_{\by_0 \to \by' \in E} J_{\by_0 \to \by'}  \text{ for every vertex $\by_0 \in V_1$}\}.
\end{equation}

\begin{remark}
\label{rmk:hat_j_g1_g}
Following~\eqref{def:hat_j_g1_g}, it is clear that $\mJ (G_1,G) \subset \hat{\mJ} (G_1,G)$.
Further, from $\{\bA_i \}^{a}_{i=1} \in \eJ(G)$ and Lemma \ref{lem:j0}, we conclude that 
\[\hat{\mJ} (G_1,G) \cap \RR^{|E_1|}_{>0} = \mJ (G_1,G).
\]
Similarly, we have $\hat{\mathcal{J}} (G_1) \cap \RR^{|E_1|}_{>0} = \mathcal{J} (G_1)$.
\end{remark}

\begin{remark}
Note that $\hat{\mathcal{J}} (G_1)$ is a linear subspace of $\RR^{|E_1|}$, while the sets $\hat{\mJ} (G_1,G)$, $\mJ (G_1,G)$ and $\mathcal{J} (G_1)$ are not linear subspaces.
\end{remark}

\begin{definition} \label{def:hpsi}
Given a weakly reversible E-graph $G_1 = (V_1, E_1)$ with its stoichiometric subspace $\mS_{G_1}$.
Consider an E-graph $G = (V, E)$ and $\bx_0\in\mathbb{R}^n_{>0}$, define the map 
\begin{equation} \label{eq:hpsi}
\hPsi: \hat{\mJ} (G_1,G) \times [(\bx_0 + \mS_{G_1} )\cap\mathbb{R}^n_{>0}] \times \RR^b \rightarrow \dK(G,G_1) \times \RR^a,
\end{equation}
such that for $(\hat{\bJ}, \bx, \bp) \in \hat{\mJ} (G_1,G) \times [(\bx_0 + \mS_{G_1} )\cap\mathbb{R}^n_{>0}] \times \mathbb{R}^b$, 
\begin{equation} \notag
\hat{\Psi} (\hat{\bJ},\bx, \bp) 
: = (\bk, \bq),
\end{equation}
where
\begin{equation} \label{def:hpsi_k}
(G, \bk) \sim (G_1, \hat{\bk}_1) \ \text{ with } \ \hat{k}_{1, \by\rightarrow \by'} = \frac{\hat{J}_{\by\rightarrow \by'}}{{\bx}^{\by}},
\end{equation} 
and
\begin{equation} \label{def:hpsi_kq}
\bp = ( \langle \bk, \bB_1 \rangle, \langle \bk, \bB_2 \rangle, \ldots, \langle \bk, \bB_b \rangle), 
\ \
\bq = ( \langle \hat{\bJ}, \bA_1 \rangle, \langle \hat{\bJ}, \bA_2 \rangle, \ldots, \langle \hat{\bJ}, \bA_a \rangle ).
\end{equation} 
\end{definition}

Recall Remark \ref{rmk:mJ_dK}, $\dK (G, G_1)$ is empty if and only if $\mJ(G_1, G)$ is empty. 
If $\mJ(G_1, G) = \dK (G, G_1) = \emptyset$, then the map $\hPsi$ is trivial.
However, we are interested in the case when $\dK (G, G_1) \neq \emptyset$, therefore we assume both $\mJ(G_1, G)$ and $\dK (G, G_1)$ are non-empty sets in the rest of the paper.

\begin{lemma}
\label{lem:hpsi_well_def}
The map $\hPsi$ in Definition \ref{def:hpsi}
is well-defined.
\end{lemma}

\begin{proof}
Consider any point $(\hbJ^*, \bx^*, \bp^*) \in \hat{\mJ} (G_1,G)\times [(\bx_0 + \mS_{G_1} )\cap\mathbb{R}^n_{>0}] \times \mathbb{R}^b$.
From Equation\eqref{def:hat_j_g1_g}, there exist a $\bJ^* = (J^*_{\by\rightarrow \by'})_{\by\rightarrow \by' \in E_1} \in \mJ (G_1,G)$ and $w^*_i \in \RR$ for $1 \leq i \leq a$, such that
\[
\hbJ^* = \bJ^* + \sum\limits^a_{i=1} w^*_i \bA_i.
\]
Since $\{ \bA_i \}^a_{i=1}$ is an orthonormal basis of the subspace $\eJ(G_1)$, we obtain 
\begin{equation} \label{eq:psi_wd_1}
(G_1, \hbJ^*) \sim (G_1, \bJ^*).
\end{equation}
From $\bJ^* \in \mJ (G_1,G) \subset \bJ (G_1)$, set
$\bk_1 = (k_{1, \by\rightarrow \by'})_{\by\rightarrow \by' \in E_1}$ with 
$k_{1, \by\rightarrow \by'} = \frac{J^*_{\by \rightarrow \by'} }{ (\bx^*)^{\by} }$. Then
\begin{equation} \label{eq:psi_wd_2}
\bk_1 \in \mK_{\RR} (G_1,G) \subset \mK(G_1).
\end{equation}
Moreover, $\bx^*$ is the complex-balanced steady state of $(G_1, \bk_1)$.
Set $\hbk_1 = (\hat{k}_{1, \by\rightarrow \by'})_{\by\rightarrow \by' \in E_1}$ with 
$\hat{k}_{1, \by\rightarrow \by'} = \frac{\hat{J}^*_{\by \rightarrow \by'} }{ (\bx^*)^{\by} }$. From Equation\eqref{eq:psi_wd_1} and Proposition \ref{prop:craciun2020efficient}, we have
\begin{equation} \label{eq:psi_wd_3}
(G_1, \bk_1) \sim (G_1, \hat{\bk}_1).
\end{equation} 

From Equation\eqref{eq:psi_wd_2}, there exists a $\bk \in \dK(G,G_1) \subset \RR^{|E|}$, such that
$(G, \bk) \sim (G_1, \bk_1)$.
Now suppose $\bp^* = (p^*_1, p^*_2, \ldots, p^*_b) \in \RR^b$, we construct the vector $\bk^* \in \RR^{|E|}$ as 
\[
\bk^* = \bk + \sum\limits^{b}_{i=1} (p^*_i - \langle \bk, \bB_i \rangle ) \bB_i.
\]
Since $\{ \bB_i \}^b_{i=1}$ is an orthonormal basis of the subspace $\mD(G)$, then for $1 \leq j \leq b$,
\begin{equation} \label{eq:k*p*}
\langle \bk^*, \bB_j \rangle 
= \langle \bk + \sum\limits^{b}_{i=1} (p^*_i - \langle \bk, \bB_i \rangle ) \bB_i, \bB_j \rangle
= \langle \bk, \bB_j \rangle + (p^*_j - \langle \bk, \bB_j \rangle ) = p^*_j.
\end{equation}
Using Lemma \ref{lem:d0}, together with
$\sum\limits^{b}_{i=1} (p^*_i - \bk \bB_i ) \bB_i \in \mD(G)$ and \eqref{eq:psi_wd_3}, we obtain
\begin{equation}
\label{eq:psi_wd_4}
(G, \bk^*) \sim (G, \bk) \sim (G_1, \hat{\bk}_1).
\end{equation}
Therefore, $\bk^*$ satisfies Equations\eqref{def:hpsi_k} and \eqref{def:hpsi_kq}.

\smallskip

\noindent Let us assume that there exists $\bk^{**} \in \dK(G,G_1)$ satisfying Equations\eqref{def:hpsi_k} and \eqref{def:hpsi_kq}, i.e.,
\[(G, \bk^{**}) \sim (G_1, \hat{\bk}_1)
\ \text{ and } \
\bp^* = ( \langle \bk^{**}, \bB_1 \rangle, \langle \bk^{**}, \bB_2 \rangle, \ldots, \langle \bk^{**}, \bB_b \rangle).
\]
This implies that $(G, \bk^{**}) \sim (G, \bk^*)$. From Lemma \ref{lem:d0}, we obtain
\[
\bk^{**} - \bk^{*} \in \mD(G).
\]
Using \eqref{eq:k*p*}, we get
\[
\langle \bk^*, \bB_j \rangle  
= \langle \bk^{**}, \bB_j \rangle = p^*_j
\ \text{ for any } \
1 \leq j \leq b.
\]
Recall that $\{ \bB_i \}^b_{i=1}$ is an orthonormal basis of $\mD(G)$. Therefore, we get 
\[
\bk^{**} = \bk^{*}.
\]
This implies that $\bk^* \in \dK(G,G_1)$ is well-defined. Moreover, from \eqref{def:hpsi_kq} we obtain 
\[
\bq^* = ( \langle \hbJ^*, \bA_1 \rangle, \langle \hbJ^*, \bA_2 \rangle, \ldots, \langle \hbJ^*, \bA_a \rangle )
\ \text{ is well-defined}.
\] 
This implies that we get
\[
\hPsi (\hbJ^*, \bx^*, \bp^*) = (\bk^*, \bq^*),
\]
and thus the map $\hPsi$ is well-defined.
\end{proof}

The following is a direct consequence of Lemma \ref{lem:hpsi_well_def}.

\begin{corollary}
\label{cor:hpsi_ss}
Consider the map $\hPsi$ in Definition \ref{def:hpsi}.
Suppose that $\hat{\Psi} (\hat{\bJ},\bx, \bp) = (\bk, \bq)$, then $\bx$ is a steady state of the system $(G, \bk)$.
\end{corollary}

\begin{proof}

It is clear that $\hat{\bJ} \in \hat{\mJ} (G_1,G)$ and $\bx \in (\bx_0 + \mS_{G_1} )\cap\mathbb{R}^n_{>0}$.
From Equation\eqref{def:hat_j_g1_g}, there exist some $\bJ^* = (J^*_{\by\rightarrow \by'})_{\by\rightarrow \by' \in E_1} \in \mJ (G_1,G)$, such that 
\[
\hbJ - \bJ^* \in \spn \{\bA_i \}^{a}_{i=1}.
\]
Using \eqref{eq:psi_wd_2} and setting $\bk_1 = (k_{1, \by\rightarrow \by'})_{\by\rightarrow \by' \in E_1}$ with $k_{1, \by\rightarrow \by'} = \frac{J^*_{\by \rightarrow \by'} }{ (\bx^*)^{\by} }$, we derive
\[
\bk_1 \in \mK_{\RR} (G_1,G),
\]
and $\bx^*$ is the complex-balanced steady state of $(G_1, \bk_1)$.
Finally, using Equations\eqref{eq:psi_wd_3} and \eqref{eq:psi_wd_4}, together with Remark \ref{rmk:de_ss}, we obtain $(G, \bk) \sim (G_1, \bk_1)$ and conclude that $\bx$ is a steady state of the system $(G, \bk)$.
\end{proof}

\begin{lemma}
\label{lem:hpsi_bijective}
The map $\hPsi$ in Definition \ref{def:hpsi}
is bijective.
\end{lemma}

\begin{proof}

First, we show the map $\hPsi$ is injective. 
Suppose two elements $(\hbJ^*, \bx^*, \bp^*)$ and $(\hbJ^{**}, \bx^{**}, \bp^{**})$ of $\hat{\mJ} (G_1,G) \times [(\bx_0 + \mS_{G_1} )\cap\mathbb{R}^n_{>0}] \times \mathbb{R}^b$ satisfy 
\[
\hPsi (\hbJ^*, \bx^*, \bp^*) = \hPsi (\hbJ^{**}, \bx^{**}, \bp^{**}) = (\bk, \bq) \in \dK(G,G_1)\times \RR^a.
\]
From \eqref{def:hat_j_g1_g}, there exist $\bJ^* = (J^*_{\by\rightarrow \by'})_{\by\rightarrow \by' \in E_1} \in \mJ (G_1,G)$ and $\bJ^{**} = (J^{**}_{\by\rightarrow \by'})_{\by\rightarrow \by' \in E_1} \in \mJ (G_1,G)$, such that
\begin{equation}
\label{eq:hpsi_bijective_1}
\hbJ^* - \bJ^* \in \spn \{ \bA_i \}^{a}_{i=1} 
\ \text{ and } \
\hbJ^{**} - \bJ^{**}
\in \spn \{ \bA_i \}^{a}_{i=1}.
\end{equation}
Then we set
$\bk^* = (k^*_{\by\rightarrow \by'})_{\by\rightarrow \by' \in E_1}$ and $\bk^{**} = (k^{**}_{\by\rightarrow \by'})_{\by\rightarrow \by' \in E_1}$ with
\[
k^*_{\by\rightarrow \by'} = \frac{J^*_{\by\rightarrow \by'}}{{(\bx^*)}^{\by}}
\ \text{ and } \
k^{**}_{\by\rightarrow \by'} = \frac{J^{**}_{\by\rightarrow \by'}}{{(\bx^*)}^{\by}}.
\]
Using Propositions\ref{prop:craciun2020efficient} and \eqref{def:hpsi_k}, we get 
\[\bk^*, \bk^{**} \in \mK_{\RR} (G_1,G)
\ \text{ and } \ 
(G, \bk) \sim (G_1, \bk^*) \sim (G_1, \bk^{**}).
\]
Moreover, two complex-balanced system $(G_1, \bk^*)$ and $(G_1, \bk^{**})$ admit steady states
\[
\bx^* \in
(\bx_0 + \mS_{G_1} )\cap\mathbb{R}^n_{>0}
\ \text{ and } \ 
\bx^{**} \in
(\bx_0 + \mS_{G_1} )\cap\mathbb{R}^n_{>0},
\ \text{respectively}.
\]
Since every complex-balanced system has a unique steady state within each invariant polyhedron and $(G_1, \bk^*) \sim (G_1, \bk^{**})$, then
\[
\bx^* = \bx^{**}.
\]
Now applying Proposition \ref{prop:craciun2020efficient} and Lemma \ref{lem:j0}, we get 
\begin{equation} 
\label{eq:hpsi_bijective_2}
(G_1, \bJ^*) \sim (G_1, \bJ^{**})
\ \text{ and } \
\bJ^{**} - \bJ^* \in \eJ(G_1).
\end{equation}
Since $\eJ(G_1) = \spn \{ \bA_i \}^{a}_{i=1}$, using \eqref{eq:hpsi_bijective_1} and \eqref{eq:hpsi_bijective_2}, we have
\begin{equation} 
\label{eq:hpsi_bijective_3}
\hbJ^{**} - \hbJ^* \in \spn \{ \bA_i \}^{a}_{i=1}.
\end{equation}
On the other hand, Equation\eqref{def:hpsi_kq} shows that
\[
\bq = ( \langle \hbJ^*, \bA_1 \rangle, \langle \hbJ^*, \bA_2 \rangle, \ldots, \langle \hbJ^*, \bA_a \rangle ) 
= ( \langle \hbJ^{**}, \bA_1 \rangle, \langle \hbJ^{**}, \bA_2 \rangle, \ldots, \langle \hbJ^{**}, \bA_a \rangle ).
\]

Since $\{\bA_i \}^{a}_{i=1}$ is an orthonormal basis of the subspace $\eJ(G)$, together with \eqref{eq:hpsi_bijective_3}, then 
\[
\hbJ^* = \hbJ^{**}.
\]
Furthermore, from \eqref{def:hpsi_kq} we obtain
\[
\bp^* =  \bp^{**} = ( \langle \bk, \bB_1 \rangle, \langle \bk, \bB_2 \rangle, \ldots, \langle \bk, \bB_b \rangle).
\]
Therefore, we show $(\bJ^*, \bx^*, \bp^*) = (\bJ^{**}, \bx^{**}, \bp^{**})$ and conclude the injectivity.

\medskip

We now show that the map $\hPsi$ is surjective.
Assume any point $(\bk, \bq) \in \dK(G,G_1)\times \RR^a$.
Since $\bk \in \dK (G, G_1)$, there exists some $\bk_1 \in \mK (G_1, G)$, such that 
\begin{equation}
\label{eq:gk_g1k1}
(G, \bk) \sim (G_1, \bk_1)
\ \text{ with } \ 
\bk_1 = (k_{1, \by\rightarrow \by'})_{\by\rightarrow \by' \in E_1}.
\end{equation}
From Theorem \ref{thm:cb}, the complex-balanced system $(G_1, \bk_1)$ has a unique steady state $\bx \in (\bx_0 + \mS_{G_1} )\cap\mathbb{R}^n_{>0}$. 
We set the flux vector $\bJ_1$ as 
\[
\bJ_1 = (J_{1, \by\rightarrow \by'})_{\by\rightarrow \by' \in E_1}
\ \text{ with } \ J_{1, \by\rightarrow \by'} = k_{1, \by\rightarrow \by'} {\bx}^{\by}.
\]
It is clear that 
$\bJ_1 \in \mJ (G_1,G)$ and the flux system $(G_1, \bJ_1)$ gives rise to the complex-balanced system $(G_1, \bk_1)$ with a  steady state $\bx \in (\bx_0 + \mS_{G_1} )\cap\mathbb{R}^n_{>0}$.
Now suppose $\bq = (q_1, q_2, \ldots, q_a)$, we construct a new flux vector $\hbJ$ as follows:
\[
\hbJ = \bJ_1 + \sum\limits^{a}_{i=1} (q_i - \langle \bJ_1, \bA_i \rangle ) \bA_i.
\]
Using the fact that $\{ \bA_i \}^a_{i=1}$ is an orthonormal basis of the subspace $\eJ(G_1)$, we can compute
\begin{equation} \notag
\langle \hbJ, \bA_i \rangle = \hat{q}_i
\ \text{ for any } \
1 \leq i \leq a.
\end{equation}
From Lemma \ref{lem:j0} and $\sum\limits^{a}_{i=1} (q_i - \langle\bJ_1 \bA_i\rangle ) \bA_i \in \eJ(G_1)$, we obtain
\[
(G, \hbJ) \sim (G_1, \bJ_1).
\]
Let $\hbk_1 = (k_{1, \by\rightarrow \by'})_{\by\rightarrow \by' \in E_1}$ with $\hat{k}_{1, \by\rightarrow \by'} = \frac{\hat{J}_{\by\rightarrow \by'}}{{\bx}^{\by}}$. From Proposition \ref{prop:craciun2020efficient} and \eqref{eq:gk_g1k1}, we have
\[
(G, \bk) \sim (G_1, \bk_1) \sim (G, \hbk_1).
\]
Finally, let $\bp = ( \langle \bk, \bB_1 \rangle, \langle \bk, \bB_2 \rangle, \ldots, \langle \bk, \bB_b \rangle)$ and derive that
\[
\hat{\Psi} (\hat{\bJ},\bx, \bp) = (\bk, \bq).
\]
Therefore, we prove the map $\hat{\Psi}$ is surjective.
\end{proof}

\begin{lemma}
\label{lem:hpsi_cts} 
The map $\hPsi$ in Definition \ref{def:hpsi} is continuous.
\end{lemma}

\begin{proof}
Consider any fixed point $(\hbJ, \bx, \bp) \in \hmJ (G_1,G)\times [(\bx_0 + \mS_{G_1} )\cap\mathbb{R}^n_{>0}] \times \mathbb{R}^b$, such that
\[
\hPsi (\hbJ, \bx, \bp) = (\bk, \bq).
\]
From \eqref{def:hpsi_kq} in Definition \ref{def:hpsi}, $\bq$ is defined as
\[
\bq = ( \langle \hat{\bJ}, \bA_1 \rangle, \langle \hat{\bJ}, \bA_2 \rangle, \ldots, \langle \hat{\bJ}, \bA_a \rangle ).
\]
It follows that $\bq$ is a continuous function of $\hbJ$.

\smallskip

Now it remains to show that $\bk$ is also a continuous function of $(\hbJ,\bx,\bq)$.
Recall \eqref{def:hpsi_k} in Definition \ref{def:hpsi}, $\bk$ is defined as
\[
(G, \bk) \sim (G_1, \hat{\bk}_1) \ \text{ with } \ \hat{k}_{1, \by\rightarrow \by'} = \frac{\hat{J}_{\by\rightarrow \by'}}{{\bx}^{\by}}.
\]
Together with \eqref{def:hpsi_kq}, we get 
\begin{equation} \label{eq:k_ct_2}
\bp = ( \langle \bk, \bB_1 \rangle, \langle \bk, \bB_2 \rangle, \ldots, \langle \bk, \bB_b \rangle),
\end{equation}
and for every vertex $\by_0 \in V \cup V_1$,
\begin{equation} \label{eq:k_ct_1}
\sum_{\by_0 \to \by \in E} k_{\by_0  \to \by} (\by - \by_0) 
= \sum_{\by_0 \to \by' \in E_1} \frac{\hat{J}_{\by_0 \rightarrow \by'}}{{\bx}^{\by_0}} (\by' - \by_0).
\end{equation}
Note that $\hbJ$ and $\bx$ are fixed, then \eqref{eq:k_ct_1} can be rewritten as
\begin{equation} \label{eq:k_ct_1_1}
\sum_{\by_0 \to \by \in E} k_{\by_0  \to \by} (\by - \by_0) 
= \text{constant}. 
\end{equation}

Assume $\bk'$ is another solution to \eqref{eq:k_ct_1_1}, then 
\[
(G, \bk) \sim (G, \bk').
\]
Using Lemma \ref{lem:d0}, we obtain that 
\[
\bk' - \bk \in \mD (G).
\]
Together with the linearity of $\mD (G)$, the solutions to \eqref{eq:k_ct_1_1} form an affine linear subspace.
Hence, the tangent space of the solution to \eqref{eq:k_ct_1_1} at $(\bJ, \bx, \bp)$ is $\mD(G)$.

Analogously, given fixed $\bp$, the solutions to \eqref{eq:k_ct_2} also form an affine linear subspace, whose tangent space at $(\bJ, \bx, \bp)$ is tangential to 
\begin{equation} \notag
\spn \{\bB_1, \bB_2, \ldots, \bB_b\} = \mD(G).
\end{equation}

This indicates that two tangent spaces at $(\bJ, \bx, \bp)$ are complementary, and thus intersect transversally~\cite{guillemin2010differential}. 
From Lemma \ref{lem:hpsi_well_def}, $\bk$ is the unique solution to \eqref{eq:k_ct_2} and \eqref{eq:k_ct_1}. Therefore, we conclude that $\bk$ as the unique intersection point (solution) of two equations \eqref{eq:k_ct_2} and \eqref{eq:k_ct_1} must vary continuously with respect to parameters $(\hbJ, \bx, \bp)$.
\end{proof}

\section{Continuity of \texorpdfstring{$\hPsi^{-1}$}{hPsi-1}}
\label{sec:continuity}

In this section, we first introduce the map $\Phi$ (see Definition \ref{def:phi}) and prove $\Phi = \hPsi^{-1}$ is well-defined. 
Then we show the map $\Phi$ is continuous, i.e. $\hPsi^{-1}$ is also continuous.

\begin{definition} \label{def:phi}
Given a weakly reversible E-graph $G_1 = (V_1, E_1)$ with its stoichiometric subspace $\mS_{G_1}$.
Consider an E-graph $G = (V, E)$ and $\bx_0\in\mathbb{R}^n_{>0}$, define the map 
\begin{equation} \label{eq:phi}
\Phi: \dK(G,G_1)\times \RR^a \rightarrow \hat{\mJ} (G_1,G) \times [(\bx_0 + \mS_{G_1} )\cap\mathbb{R}^n_{>0}] \times \RR^b,
\end{equation}
such that for $(\bk, \bq) \in \dK(G,G_1)\times \RR^a$, 
\begin{equation} \notag
\Phi (\bk, \bq) := (\hat{\bJ},\bx, \bp),
\end{equation}
where $\bx \in (\bx_0 + \mS_{G_1} )\cap\mathbb{R}^n_{>0}$ is the steady state of $(G, \bk)$, and
\begin{equation} \label{def:phi_k}
(G, \bk) \sim (G_1, \hat{\bk}_1) \ \text{ with } \ \hat{k}_{1, \by\rightarrow \by'} = \frac{\hat{J}_{\by\rightarrow \by'}}{{\bx}^{\by}},
\end{equation} 
and
\begin{equation} \label{def:phi_kq}
\bp = ( \langle \bk, \bB_1 \rangle, \langle \bk, \bB_2 \rangle, \ldots, \langle \bk, \bB_b \rangle), 
\ \
\bq = ( \langle \hat{\bJ}, \bA_1 \rangle, \langle \hat{\bJ}, \bA_2 \rangle, \ldots, \langle \hat{\bJ}, \bA_a \rangle ).
\end{equation} 
\end{definition}

\medskip

\begin{lemma}
\label{lem:phi_wd}
The map $\Phi$ in Definition \ref{def:phi} is well-defined, and $\Phi = \hPsi^{-1}$ is bijective.
\end{lemma}

\begin{proof}

Assume any point $(\bk^*, \bq^*) \in \dK(G,G_1)\times \RR^a$.
There exists $\bk_1 \in \mK_{\RR} (G_1,G)$ satisfying
\begin{equation} \label{eq:phi_wd_1}
(G, \bk^*) \sim (G_1, \bk_1).
\end{equation}
From Theorem \ref{thm:cb}, $(G_1, \bk_1)$ has a unique steady state $\bx^* \in (\bx_0 + \mS_{G_1} )\cap\mathbb{R}^n_{>0}$.
Further, Remark \ref{rmk:de_ss} shows that $(G, \bk^*)$ and $(G_1, \bk_1)$ share the same steady states, thus $\bx^* \in (\bx_0 + \mS_{G_1} )\cap\mathbb{R}^n_{>0}$ is also the unique steady state of $(G, \bk^*)$, i.e. $\bx^*$ is well-defined.
Moreover, from \eqref{def:phi_kq} we obtain 
\begin{equation}
\label{eq:phi_wd_2}
\bp^* = ( \langle \bk^*, \bB_1 \rangle, \langle \bk^*, \bB_2 \rangle, \ldots, \langle \bk^*, \bB_b \rangle),
\end{equation}
which is well-defined.

Since $\bk_1 \in \mK_{\RR} (G_1,G)$, then $(G_1, \bk_1)$ and its steady state $\bx^*$ give rise to the complex-balanced flux system $(G_1, \bJ^*)$, such that
\[
\bJ^* = (J^*_{\by\rightarrow \by'})_{\by\rightarrow \by' \in E_1} \in \mJ (G_1,G)
\ \text{ with } \ 
J^*_{\by\rightarrow \by'} = k_{1, \by\rightarrow \by'} (\bx^*)^{\by}.
\]
Suppose $\bq^* = (q^*_1, q^*_2, \ldots, q^*_a) \in \RR^a$, we construct the vector $\hbJ^* \in \RR^{|E|}$ as 
\[
\hbJ^* = \bJ^* + \sum\limits^a_{i=1} (q^*_i - \langle \bJ^*, \bA_i \rangle ) \bA_i \in \hat{\mJ} (G_1,G).
\]
Note that $\{ \bA_i \}^a_{i=1}$ is an orthonormal basis of $\eJ(G_1)$, together with Lemma \ref{lem:j0}, we obtain
\begin{equation} \notag
\bq^* = ( \langle \hbJ^*, \bA_1 \rangle, \langle \hbJ^*, \bA_2 \rangle, \ldots, \langle \hbJ^*, \bA_a \rangle )
\ \text{ and } \
(G_1, \hbJ^*) \sim (G_1, \bJ^*).
\end{equation}
Using Proposition \ref{prop:craciun2020efficient} and \eqref{eq:phi_wd_1}, we set 
$\hbk_1 = (\hat{k}_{1, \by\rightarrow \by'})_{\by\rightarrow \by' \in E_1}$ with
$\hat{k}_{1, \by\rightarrow \by'} = \frac{\hat{J}^*_{\by\rightarrow \by'}}{{(\bx^*)}^{\by}}$ and derive 
\begin{equation} \notag
(G_1, \hat{\bk}_1) \sim (G_1, \bk_1) \sim (G, \bk^*).
\end{equation}
Together with \eqref{eq:phi_wd_2}, we conclude that $(\hbJ^*, \bx^*, \bp^*)$ satisfies \eqref{def:phi_k} and \eqref{def:phi_kq}.

Now suppose there exists another $(\hbJ^{**}, \bx^{**}, \bp^{**}) \in \hat{\mJ} (G_1,G)\times [(\bx_0 + \mS_{G_1} )\cap\mathbb{R}^n_{>0}] \times \mathbb{R}^b$, which also satisfies \eqref{def:phi_k} and \eqref{def:phi_kq}.
From Definition \ref{def:hpsi}, we deduce
\begin{equation} \notag
\hPsi (\hbJ^*, \bx^*, \bp^*) = \hPsi (\hbJ^{**}, \bx^{**}, \bp^{**}) = (\bk^*, \bq^*).
\end{equation}
Since $\hPsi$ is proved to be bijective in Lemma \ref{lem:hpsi_bijective}, then
\begin{equation} \notag
(\hbJ^*, \bx^*, \bp^*) = (\hbJ^{**}, \bx^{**}, \bp^{**}).
\end{equation}
Thus, we conclude that $\Phi$ is well-defined.

\smallskip

Next, for any $(\hbJ, \bx, \bp) \in \hat{\mJ} (G_1,G)\times [(\bx_0 + \mS_{G_1} )\cap\mathbb{R}^n_{>0}] \times \mathbb{R}^b$, suppose that
\begin{equation} 
\label{eq:phi_wd_3}
\hPsi (\hbJ, \bx, \bp) = (\bk, \bq) \in \dK(G,G_1)\times \RR^a.
\end{equation}
From Definition \ref{def:hpsi} and Corollary \ref{cor:hpsi_ss}, together with \eqref{def:phi_k} and \eqref{def:phi_kq}, we have
\begin{equation} 
\label{eq:phi_wd_4}
\Phi (\bk, \bq) = (\hbJ, \bx, \bp).
\end{equation}
This implies $\Phi = \hPsi^{-1}$.
Recall that $\hPsi$ is bijective, thus its inverse $\hPsi^{-1}$ is well-defined and bijective. Therefore, we prove the lemma.
\end{proof}

\begin{lemma}
\label{lem:inverse_cts_q}
Consider the map $\Phi$ in Definition \ref{def:phi}, suppose any fixed $\bk \in \dK(G,G_1)$ and $\bq_1, \bq_2 \in \RR^a$, then 
\begin{equation}
\label{eq:inverse_cts_q_1}
\Phi (\bk, \bq_1) - \Phi (\bk, \bq_2)
= \left(\sum\limits^{a}_{i=1} \varepsilon_i \bA_i, \mathbf{0}, \mathbf{0}\right),
\end{equation}
where $\bq_1 - \bq_2 := (\varepsilon_1, \varepsilon_2, \ldots, \varepsilon_a) \in \RR^a$.
\end{lemma}

\begin{proof}

Given fixed $\bk \in \dK(G,G_1)$, consider any $\bq \in \RR^a$, such that
\begin{equation} \notag
\Phi (\bk, \bq) = (\hat{\bJ},\bx, \bp).
\end{equation}
From Definition \ref{def:phi}, $\bx \in (\bx_0 + \mS_{G_1} )\cap\mathbb{R}^n_{>0}$ is the steady state of $(G, \bk)$. Further, we have
\begin{equation} 
\label{eq:inverse_cts_q_3}
(G, \bk) \sim (G_1, \hat{\bk}_1) \ \text{ with } \ \hat{k}_{1, \by\rightarrow \by'} = \frac{\hat{J}_{\by\rightarrow \by'}}{{\bx}^{\by}},
\end{equation} 
and
\begin{equation} 
\label{eq:inverse_cts_q_4}
\bp = ( \langle \bk, \bB_1 \rangle, \langle \bk, \bB_2 \rangle, \ldots, \langle \bk, \bB_b \rangle), 
\ \
\bq = ( \langle \hat{\bJ}, \bA_1 \rangle, \langle \hat{\bJ}, \bA_2 \rangle, \ldots, \langle \hat{\bJ}, \bA_a \rangle ).
\end{equation}

\smallskip

Now consider any vector $\bep = (\varepsilon_1, \varepsilon_2, \ldots, \varepsilon_a) \in \RR^a$, it follows that \eqref{eq:inverse_cts_q_1} is equivalent to show the following:
\begin{equation}
\label{eq:inverse_cts_q_2}
\Phi (\bk, \bq + \bep) = (\hat{\bJ} + \sum\limits^{a}_{i=1} \varepsilon_i \bA_i,\bx, \bp).
\end{equation}

Suppose $\Phi (\bk, \bq + \bep) = (\hbJ^{\bep}, \bx^{\bep}, \bp^{\bep})$. 
From Definition \ref{def:phi} and Lemma \ref{lem:phi_wd}, $\bx^{\bep}$ is the unique steady state of $(G, \bk)$ in the invariant polyhedron $ (\bx_0 + \mS_{G_1} )\cap\mathbb{R}^n_{>0}$. 
Recall that $\bx \in (\bx_0 + \mS_{G_1} )\cap\mathbb{R}^n_{>0}$ is also the steady state of $(G, \bk)$, thus we have
\begin{equation}
\label{eq:inverse_cts_q_6}
\bx = \bx^{\bep}.
\end{equation}
Since $\hat{\bJ} \in \hmJ (G_1,G)$ and $\{ \bA_i \}^a_{i=1}$ is an orthonormal basis of $\eJ(G_1)$, we get
\[
(G_1, \hat{\bJ}) \sim (G_1, \hat{\bJ} + \sum\limits^{a}_{i=1} \varepsilon_i \bA_i).
\]
Using Proposition \ref{prop:craciun2020efficient} and \eqref{eq:inverse_cts_q_3}, by setting $\hat{J}_{\by\rightarrow \by'} + \sum\limits^{a}_{i=1} \varepsilon_i \bA_{i, \by\rightarrow \by'} = \hat{k}^{\bep}_{1, \by\rightarrow \by'} \bx^{\by}$, we obtain  
\begin{equation} 
\label{eq:inverse_cts_q_5}
(G_1, \hat{\bk}^{\bep}_1) \sim (G_1, \hat{\bk}_1) \sim (G, \bk).
\end{equation}
Under direct computation, for $1 \leq i \leq a$,
\begin{equation} \notag
\langle \hat{\bJ} + \sum\limits^{a}_{i=1} \varepsilon_i \bA_i, \bA_i \rangle
= \langle \hat{\bJ}, \bA_i \rangle + \langle \sum\limits^{a}_{i=1} \varepsilon_i \bA_i, \bA_i \rangle = \langle \hat{\bJ}, \bA_i \rangle + \varepsilon_i.
\end{equation}
From Lemma \ref{lem:phi_wd} and \eqref{eq:inverse_cts_q_5}, we get
\begin{equation}
\label{eq:inverse_cts_q_7}
\hbJ^{\bep} = \hat{\bJ} + \sum\limits^{a}_{i=1} \varepsilon_i \bA_i.
\end{equation}
Finally, from Definition \ref{def:phi} and \eqref{eq:inverse_cts_q_4}, it is clear that
\begin{equation}
\label{eq:inverse_cts_q_8}
\bp^{\bep} = ( \langle \bk, \bB_1 \rangle, \langle \bk, \bB_2 \rangle, \ldots, \langle \bk, \bB_b \rangle ) = \bp.
\end{equation}
Combining Equations~\eqref{eq:inverse_cts_q_6}, \eqref{eq:inverse_cts_q_7} and \eqref{eq:inverse_cts_q_8}, we prove \eqref{eq:inverse_cts_q_2}.
\end{proof}

Here we present Proposition \ref{prop:inverse_cts_k}, which is the key for the continuity of $\hPsi^{-1}$.

\begin{proposition}
\label{prop:inverse_cts_k}
Consider the map $\Phi$ in Definition \ref{def:phi} and any fixed $\bq \in \RR^a$, then $\Phi (\cdot, \bq)$ is continuous with respect to $\bk$.
\end{proposition}

To prove Proposition~\ref{prop:inverse_cts_k}, we need to show Lemmas \ref{lem:key_1} - \ref{lem:key_3} and Proposition \ref{lem:key_4}. 
The following is the overview of the process. First, Lemma \ref{lem:key_1} shows that if two reaction rate vectors in $\dK (G, G_1)$ are close enough, then there exist two reaction rate vectors (dynamically equivalent respectively) in $\mK (G_1, G_1)$ such that their distance can be controlled. 
Second, in Lemma \ref{lem:key_2} we show that given a complex-balanced rate vector $\bk_1 \in \mK (G_1)$, there exists a neighborhood around $\bk_1$ of $\RR^{E_1}_{>0}$, in which the steady states of the system associated with the rate constants vary continuously. Combining Lemma \ref{lem:key_1} with \ref{lem:key_2}, we prove in Lemma \ref{lem:key_3} that given a reaction rate vector $\bk \in \dK (G, G_1)$, there exists an open neighborhood $\bk \in U \subset \RR^{E}$, such that the steady states of the system associated with the rate vectors in $U$ vary continuously.
Finally, in Proposition \ref{lem:key_4} we prove that given a complex-balanced rate vector $\bk^* \in \mK (G_1, G_1)$, for any sequence $\bk_i \to \bk^*$ in $\mK (G_1, G_1)$, there exists another sequence of reaction rate vectors (dynamically equivalent respectively) $\hbk_i \to \bk^*$ in $\RR^{E_1}$, and all associated fluxes from reaction rate vectors have the same projections on $\eJ (G_1)$. 

\medskip

\begin{lemma} \label{lem:key_1}

Let $\bk \in \dK (G,G_1)$. Then we have the following:
\begin{enumerate}[label=(\alph*)]
\item There exists $\bk_1 \in \mK (G_1)$ satisfying $(G, \bk) \sim (G_1, \bk_1)$.
\item There exist constants $\varepsilon = \varepsilon (\bk) > 0$ and $C = C (\bk) > 0$, such that for any $\hbk \in \dK (G,G_1)$ with $\| \hbk - \bk \| \leq \varepsilon$, there exists $\hbk_1 \in \mK (G_1,G_1)$ that satisfies
\begin{enumerate}[label=(\roman*)]
\item $\|\hbk_1 - \bk_1 \| \leq C \varepsilon $.
\item $(G,\hbk) \sim (G_1, \hbk_1)$.
\end{enumerate}
\end{enumerate}
\end{lemma}

\begin{proof}
For part $(a)$, from Definitions \ref{def:mas_realizable} and \ref{def:de_realizable}, given $\bk \in \dK (G,G_1)$, the system $(G, \bk)$ is disguised toric on $G_1$, that is, there exists $\bk_1 \in \mK_{\RR} (G_1, G) \subset \mK (G_1)$ with $(G, \bk) \sim (G_1, \bk_1)$. 

\smallskip

Now we prove part $(b)$.\\
\textbf{Step 1: } 
Let $\by \in G \cup G_1$ be a fixed vertex and consider the following vector space:
\begin{equation} \notag
W_{\by} = \spn \{ \by' - \by: \by \rightarrow \by' \in G_1 \}.
\end{equation}
Let $d(\by) = \dim (W_{\by})$. Then there exists an orthogonal basis of $W_{\by}$ denoted by: 
\begin{equation} \label{eq:key_1_1}
\{ \bw_1, \bw_2, \ldots, \bw_{d (\by)} \}.
\end{equation}
For each $\bw_i$ in \eqref{eq:key_1_1}, there exist positive $\{ c_{i, \by \rightarrow \by'} \}_{\by \rightarrow \by' \in G_1}$, that satisfy
\begin{equation} \label{eq:key_1_2}
\bw_i = \sum\limits_{\by \rightarrow \by' \in G_1} c_{i, \by \rightarrow \by'} (\by' - \by).
\end{equation}
Let $\hbk \in \dK (G,G_1)$. From Definition \ref{def:de_realizable}, 
$\sum\limits_{\by \rightarrow \tilde{\by} \in G} \hbk_{\by \rightarrow \tilde{\by}} (\tilde{\by} - \by)$ is realizable on $G_1$ at the vertex $\by \in G \cup G_1$. This implies that

\begin{equation} \label{eq:key_1_3}
\sum\limits_{\by \rightarrow \tilde{\by} \in G} \hbk_{\by \rightarrow \tilde{\by}} (\tilde{\by} - \by) \in W_{\by}.
\end{equation}
Since $\bk \in \dK (G,G_1)$, together with Equation~\eqref{eq:key_1_3}, we obtain
\begin{equation} 
\label{eq:key_1_Delta}
\Delta_{\by} (\hbk, \bk) := \sum\limits_{\by \rightarrow \tilde{\by} \in G} ( \hbk_{\by \rightarrow \tilde{\by}} - \bk_{\by \rightarrow \tilde{\by}}) (\tilde{\by} - \by) \in W_{\by}.
\end{equation}

Assume that $\| \hbk - \bk \| \leq \varepsilon$. 
Consider all reaction vectors in $G$ and let $m = \max\limits_{\by \rightarrow \tilde{\by} \in G} \| \tilde{\by} - \by \|$, then there exists a constant $C_1 = m |E|$, such that 
\[
\| \Delta_{\by} (\hbk, \bk) \| 
\leq \sum\limits_{\by \rightarrow \tilde{\by} \in G} m \varepsilon = C_1 \varepsilon.
\]
On the other side, from \eqref{eq:key_1_1}, $\Delta_{\by} (\hbk, \bk)$ can be expressed as
\begin{equation} \label{eq:key_1_4}
\Delta_{\by} (\hbk, \bk) = \sum\limits^{d(\by)}_{i=1} \delta_i \bw_i
\ \text{ with } \
\delta_i \in \RR.
\end{equation}
Using \eqref{eq:key_1_4} and the orthogonal basis in \eqref{eq:key_1_1},  for any $1 \leq i \leq d (\by)$,
\begin{equation}
\label{eq:key_1_5}
| \delta_i | \leq \| \Delta_{\by} (\hbk, \bk) \| \leq C_1 \varepsilon.
\end{equation}
Inputting \eqref{eq:key_1_2} into \eqref{eq:key_1_4}, we get
\begin{equation}
\label{eq:key_1_6}
\Delta_{\by} (\hbk, \bk) 
= \sum\limits^{d(\by)}_{i=1} \delta_i \big( \sum\limits_{\by \rightarrow \by' \in G_1} c_{i, \by \rightarrow \by'} (\by' - \by) \big)
= \sum\limits_{\by \rightarrow \by' \in G_1} \big( \sum\limits^{d(\by)}_{i=1} \delta_i c_{i, \by \rightarrow \by'} \big) (\by' - \by).
\end{equation}
From \eqref{eq:key_1_5} and \eqref{eq:key_1_6}, there exists a constant $C_2$, such that for any $\by \rightarrow \by' \in G_1$,
\begin{equation}
\label{eq:key_1_7}
\big| \hat{c}_{\by \rightarrow \by'} := \sum\limits^{d(\by)}_{i=1} \delta_i c_{i, \by \rightarrow \by'} \big| 
\leq C_2 \varepsilon.
\end{equation}
Then we construct $\hbk_1$ as follows: 
\begin{equation}
\label{eq:key_1_8}
\hbk_{1, \by \rightarrow \by'} := \bk_{1, \by \rightarrow \by'} + \hat{c}_{\by \rightarrow \by'}
\ \text{ for any } \
\by \rightarrow \by' \in G_1.
\end{equation} 
Consider all reaction vectors in $G_1$, together with \eqref{eq:key_1_7}, we derive 
\begin{equation} 
\label{eq:key_1_estimate}
\| \hbk_1 - \bk_1 \| 
\leq \sum\limits_{\by \rightarrow \by' \in G_1} |\hat{c}_{\by \rightarrow \by'}|
\leq \sum\limits_{\by \rightarrow \by' \in G_1} C_2 \varepsilon \leq C_2 |E_1| \varepsilon.
\end{equation}

Similarly, we can go through all vertices in $G \cup G_1$, and take the above steps to update $\hbk_1$.
For every vertex, we can derive an estimate similar to \eqref{eq:key_1_estimate}.
Collecting the estimates on all vertices, we can find a constant $C$, such that
\[
\| \hbk_1 - \bk_1 \| \leq C \varepsilon
\ \text{ for any } \
\| \hbk - \bk \| \leq \varepsilon.
\]

\textbf{Step 2: } 
We claim that there exists a sufficiently small constant $\varepsilon = \varepsilon (\bk) > 0$, such that for any $\hbk$  with $\| \hbk - \bk \| \leq \varepsilon$, then $\hbk_1$ defined in \eqref{eq:key_1_8} satisfies
\begin{equation}
\label{eq:key_1_claim}
(G, \hbk) \sim (G_1, \hbk_1)
\ \text{ and } \
\hbk_1 \in \mK (G_1,G_1).
\end{equation}
Recall \eqref{eq:key_1_3} and \eqref{eq:key_1_Delta}, at vertex $\by \in G \cup G_1$,
\begin{equation}
\label{eq:key_1_9}
\Delta_{\by} (\hbk, \bk)
= \sum\limits_{\by \rightarrow \tilde{\by} \in G} \hbk_{\by \rightarrow \tilde{\by}} (\tilde{\by} - \by)
- \sum\limits_{\by \rightarrow \tilde{\by} \in G} \bk_{\by \rightarrow \tilde{\by}} (\tilde{\by} - \by).
\end{equation}
On the other hand, from \eqref{eq:key_1_6}-\eqref{eq:key_1_8}, at vertex $\by \in G \cup G_1$,
\begin{equation}
\label{eq:key_1_10}
\Delta_{\by} (\hbk, \bk)
= \sum\limits_{\by \rightarrow \by' \in G_1} \hbk_{1, \by \rightarrow \by'} (\by' - \by)
- \sum\limits_{\by \rightarrow \by' \in G_1} \bk_{1, \by \rightarrow \by'} (\by' - \by).
\end{equation}
Note that $(G, \bk) \sim (G_1, \bk_1)$ implies that, at vertex $\by \in G \cup G_1$,
\[
\sum\limits_{\by \rightarrow \tilde{\by} \in G} \bk_{\by \rightarrow \tilde{\by}} (\tilde{\by} - \by) = \sum\limits_{\by \rightarrow \by' \in G_1} \bk_{1, \by \rightarrow \by'} (\by' - \by).
\]
Together with \eqref{eq:key_1_9} and \eqref{eq:key_1_10}, we have, at vertex $\by \in G \cup G_1$,
\begin{equation}
\sum\limits_{\by \rightarrow \tilde{\by} \in G} \hbk_{\by \rightarrow \tilde{\by}} (\tilde{\by} - \by)
= \sum\limits_{\by \rightarrow \by' \in G_1} \hbk_{1, \by \rightarrow \by'} (\by' - \by).
\end{equation}
Hence, we derive $(G, \hbk) \sim (G_1, \hbk_1)$. Moreover, since $\hbk \in \dK (G,G_1)$, there exists $\hbk^* \in \mK (G_1)$ with $(G, \hbk) \sim (G_1, \hbk^*)$, and thus
\[
(G_1, \hbk_1) \sim (G_1, \hbk^*).
\]

Recall that $\bk_1 \in \mK (G_1) \subset \RR^{E_1}_{>0}$, together with \eqref{eq:key_1_estimate}, there must exist a constant $\varepsilon = \varepsilon (\bk) > 0$, such that for any $\hbk$  with $\| \hbk - \bk \| \leq \varepsilon$, we have $\hbk_1 \in \RR^{E_1}_{>0}$.
Therefore, we obtain $\hbk_1 \in \mK (G_1,G_1)$ and prove the claim.
\end{proof}

\begin{lemma} \label{lem:key_2}

Suppose $\bx_0 \in \mathbb{R}^n_{>0}$ and $\bk_1 \in \mK (G_1)$, then there exists an open set $U \subset \RR^{E_1}_{>0}$ containing $\bk_1$, such that there exists a unique continuously differentiable function
\begin{equation} \label{lem:key_2_1}
T : U \rightarrow (\bx_0 + \mS_{G_1} )\cap\mathbb{R}^n_{>0}.
\end{equation}
such that for any $\hbk \in U$,
\begin{equation} \label{lem:key_2_2}
T (\hbk) = \hbx,
\end{equation}
where $\hbx \in (\bx_0 + \mS_{G_1} )\cap\mathbb{R}^n_{>0}$ is the steady state of $(G_1, \hbk)$.
\end{lemma}

\begin{proof}

Given $\bx_0 \in \mathbb{R}^n_{>0}$ and $\bk_1 \in \mK (G_1)$, Theorem \ref{thm:cb} shows the system $(G_1, \bk_1)$ has a unique steady state $\bx^* \in (\bx_0 + \mS_{G_1}) \cap \mathbb{R}^n_{>0}$.
Consider the system $(G_1, \bk_1)$ as follows:
\begin{equation} \label{eq:key_2_0}
\frac{d\bx}{dt} = \bf (\bk_1, \bx) := (\bf_1, \bf_2, \ldots, \bf_n)^{\intercal}
= \sum_{\by_i \rightarrow \by_j \in E_1} k_{1, \by_i \rightarrow \by_j} \bx^{\by_i}(\by_j - \by_i).
\end{equation}
Suppose $\dim (\mS_{G_1}) = s \leq n$. This implies that there exist exactly $s$ linearly independent components among $\bf (\bk_1, \bx)$.
Without loss of generality, we assume that $\{\bf_1, \ldots, \bf_s \}$ are linearly independent components, and every $\bf_i$ with $s+1 \leq i \leq n$ can be represented as a linear combination of $\{\bf_i \}^{s}_{i=1}$.

Using Theorem~\ref{thm:jacobian}, we obtain that
\begin{equation} \notag
\ker \Big( \big[ \frac{\partial \bf_i}{ \partial \bx_j} \big]_{1 \leq i, j \leq n}  \big|_{\bx = \bx^*} \Big) = \mS^{\perp}_{G_1}.
\end{equation}
Together with the linear dependence among $\{ \bf_i (\bx) \}^{n}_{i=1}$, we derive 
\begin{equation} \label{eq:key_2_1}
\ker \Big( \big[ \frac{\partial \bf_i}{ \partial \bx_j} \big]_{1 \leq i \leq s, 1 \leq j \leq n}  \big|_{\bx = \bx^*} \Big) = \mS^{\perp}_{G_1}.
\end{equation}
Consider the orthogonal complement $\mS^{\perp}_{G_1}$ to the stoichiometric subspace in $\mathbb{R}^n$, which admits an orthonormal basis given by
\[
\{\bv_1, \bv_2, \ldots, \bv_{n-s} \}.
\]

Now we construct a system of $n$ equations $\bg (\bk, \bx) = (\bg_1, \bg_2, \ldots, \bg_n )^{\intercal}$ as follows:
\begin{equation}
\label{eq:key_2_2}
\bg_i (\bk, \bx) =
\begin{cases}
\bf_i (\bk, \bx), 
& \text{ for } 1 \leq i \leq s, \\[5pt]
\bx \cdot \bv_{i-s} - \bx_0 \cdot \bv_{i-s}, & \text{ for } s+1 \leq i \leq n.
\end{cases}
\end{equation}
From \eqref{eq:key_2_0}, we can check that $\bg (\bk, \bx) = \mathbf{0}$ if and only if $\bx \in \bx_0 + \mS_{G_1}$ is the steady state of the system $(G_1, \bk)$.
Thus, $(\bk_1, \bx^*)$ can be considered as a solution to $\bg (\bk, \bx) = \mathbf{0}$, that is, $\bg (\bk_1, \bx^*) = \mathbf{0}$.

Computing the Jacobian matrix of $\bg (\bk, \bx)$ as in Equation~\eqref{eq:key_2_2}, we get
\begin{equation} \notag
\mathbf{J}_{\bg, \bx} =
\begin{pmatrix}
\big[ \frac{\partial \bf_i}{ \partial \bx_j} \big]_{1 \leq i \leq s, 1 \leq j \leq n} \\[5pt]
\bv_1 \\
\ldots \\
\bv_{n-s}
\end{pmatrix}.
\end{equation}
From~\eqref{eq:key_2_1}, we have 
\[
\ker \big( \mathbf{J}_{\bg, \bx} |_{\bk = \bk_1, \bx = \bx^*} \big) \subseteq \mS^{\perp}_{G_1}.
\]
Since the last $n-s$ rows of $\mathbf{J}_{\bg} (\bx)$, $\{\bv_1, \bv_2, \ldots, \bv_{n-s} \}$, is a orthonormal basis of $\mS^{\perp}_{G_1}$, we derive
\begin{equation} \label{eq:key_2_3}
\det \big( \mathbf{J}_{\bg, \bx} |_{\bk = \bk_1, \bx = \bx^*} \big) \neq 0.
\end{equation}
Hence, the Jacobian matrix $\mathbf{J}_{\bg, \bx}$ is invertible at $(\bk, \bx) = (\bk_1, \bx^*)$.
Further, note that
$\bg (\bk, \bx)$ is continuously differentiable. 
Using the implicit function theorem, for any $\hbk \in U$, we have
\begin{equation} \notag
T (\hbk) = \hbx,
\end{equation}
where $\hbx \in (\bx_0 + \mS_{G_1} )\cap\mathbb{R}^n_{>0}$ is the steady state of $(G_1, \hbk)$.

\end{proof}

\begin{lemma}
\label{lem:key_3}

Suppose $\bx_0\in\mathbb{R}^n_{>0}$ and $\bk \in \dK (G,G_1)$, then there exists an open set $U \subset \dK (G,G_1)$ containing $\bk$, such that there exists a unique continuous function
\begin{equation} \label{eq:key_3_1}
h : U \rightarrow (\bx_0 + \mS_{G_1} )\cap\mathbb{R}^n_{>0}.
\end{equation}
such that for any $\hbk \in U$,
\begin{equation} \label{eq:key_3_2}
h (\hbk) = \hbx,
\end{equation}
where $\hbx \in (\bx_0 + \mS_{G_1} )\cap\mathbb{R}^n_{>0}$ is the steady state of $(G, \hbk)$.
\end{lemma}

\begin{proof}

Given $\bk \in \dK (G, G_1)$ and $\bx_0 \in \mathbb{R}^n_{>0}$, there exists 
$\bk_1 \in \mK (G_1)$ such that 
\[
(G, \bk) \sim (G_1, \bk_1).
\]
Theorem \ref{thm:cb} shows the system $(G_1, \bk_1)$ has a unique steady state $\bx^* \in (\bx_0 + \mS_{G_1}) \cap \mathbb{R}^n_{>0}$.
Since $(G, \bk) \sim (G_1, \bk_1)$, $\bx^* \in (\bx_0 + \mS_{G_1}) \cap \mathbb{R}^n_{>0}$ is also the unique steady state of the system $(G, \bk)$.

Analogously, for any $\hbk \in \dK (G,G_1)$, it has a unique steady state of the system $(G, \hbk)$ in $(\bx_0 + \mS_{G_1}) \cap \mathbb{R}^n_{>0}$.
Thus, the function $h$ in \eqref{eq:key_3_1}-\eqref{eq:key_3_2} is well-defined.
It remains to prove that there exists an open set $U \subset \dK (G, G_1)$ containing $\bk$ and $h$ is continuous with respect to the domain $U$.

From Lemma~\ref{lem:key_2}, there exists an open set $U_1 \subset \RR^{E_1}_{>0}$ containing $\bk_1$, such that there exists a unique continuously differentiable function
\begin{equation} \label{eq:key_3_4}
T : U_1 \rightarrow (\bx_0 + \mS_{G_1} )\cap\mathbb{R}^n_{>0}.
\end{equation}
such that for any $\hbk \in U_1$,
\begin{equation} \notag
T (\hbk) = \hbx,
\end{equation}
where $\hbx \in (\bx_0 + \mS_{G_1} )\cap\mathbb{R}^n_{>0}$ is the steady state of $(G_1, \hbk)$.
Using \eqref{eq:key_3_4}, we can find a constant $\varepsilon_1 = \varepsilon_1 (\bk)$ such that
\begin{equation}
\label{eq:key_3_B}
B = \{ \bk^* \in \RR^{E_1}_{>0}: \|\bk^* - \bk_1 \| \leq \varepsilon_1 \} \subseteq U_1.
\end{equation}
Hence, it is clear that $T$ is continuous with respect to the domain $B$.

On the other hand, from Lemma \ref{lem:key_1}, there exist $\varepsilon = \varepsilon (\bk) > 0$ and $C = C (\bk) > 0$, such that for any $\hbk \in \dK (G,G_1)$ with $\| \hbk - \bk \| \leq \varepsilon$, there exists $\hbk_1 \in \mK (G_1,G_1)$ 
satisfying
\begin{equation} \label{eq:key_3_3}
\|\hbk_1 - \bk_1 \| \leq C \varepsilon
\ \text{ and } \
(G,\hbk) \sim (G_1, \hbk_1).
\end{equation}
Now pick $\varepsilon_2 = \min ( \varepsilon, \varepsilon_1 / C)$, and consider the following set:
\begin{equation} \notag
U := \{ \bk^* \in \RR^{E}_{>0}: \|\bk^* - \bk \| < \varepsilon_2 \} 
\ \cap \ \dK (G,G_1).
\end{equation}
Using~\eqref{eq:key_3_3}, we have that for any $\bk^* \in U$, there exists $\bk^*_1 \in \mK (G_1,G_1)$ such that
\begin{equation} \label{eq:key_3_5}
\| \bk^*_1 - \bk_1 \| \leq C \varepsilon_2 = \varepsilon_1
\ \text{ and } \
(G, \bk^*) \sim (G_1, \bk^*_1).
\end{equation}
From \eqref{eq:key_3_B}, this shows that $\bk^*_1 \in B$.
Further, from \eqref{eq:key_3_4} and \eqref{eq:key_3_3}, we obtain
\[
h (\bk^*) = T (\bk^*_1)
\]
Since $T$ is continuous with respect to the domain $B$, together with \eqref{eq:key_3_5} and $\bk^*_1 \in B$, we conclude that $h$ is continuous on $U$.
\end{proof}

\begin{proposition}
\label{lem:key_4}

Suppose $\bx_0 \in \RR^n_{>0}$ and $\bk^* \in \mK (G_1) \subset \mK (G_1,G_1)$. For any $\bk \in \mK (G_1,G_1)$, then we have the following:
\begin{enumerate}[label=(\alph*)]
\item The system $(G_1, \bk^*)$ has a unique steady state $\bx^* \in (\bx_0 + \mS_{G_1}) \cap \mathbb{R}^n_{>0}$.

\item The system $(G_1, \bk)$ has a unique steady state $\bx \in (\bx_0 + \mS_{G_1}) \cap \mathbb{R}^n_{>0}$.

\item 
Consider the steady state $\bx^*$ in part $(a)$ and $\bx$ obtained in part $(b)$. Then there exists a unique $\hbk \in \RR^{E_1}$, such that
\begin{enumerate}[label=(\roman*)]
\item \label{lem:key_4_a} 
$(G_1, \bk) \sim (G_1, \hbk)$.

\item\label{lem:key_4_b} 
$\hbJ := (\hat{k}_{\by \to \by'} \bx^{\by})_{\by \to \by' \in E_1} \in \hat{\mathcal{J}} (G_1)$.

\item \label{lem:key_4_c} 
$\langle \hbJ, \bA_i \rangle = \langle \bJ^*, \bA_i \rangle$ for any $1 \leq i \leq a$, where $\bJ^* := (k^*_{\by \to \by'} (\bx^*)^{\by})_{\by \to \by' \in E_1}$.

\end{enumerate}

\item  
 For any sequence $\{ \bk_i \}^{\infty}_{i = 1}$ in $\mK (G_1,G_1)$ converging to $\bk^*$, there exist a unique corresponding sequence $\{ \hbk_i \}^{\infty}_{i = 1}$ obtained from part $(c)$. 
Moreover, the sequence $\{ \hbk_i \}^{\infty}_{i = 1}$ satisfies
\begin{equation}
\notag
\hbk_i \to \bk^* 
\ \text{ as } \
i \to \infty.
\end{equation}

\end{enumerate}
\end{proposition}

\begin{proof}

For part (a), since $\bk^* \in \mK (G_1)$, Theorem \ref{thm:cb} shows that the system $(G_1, \bk^*)$ has a unique steady state $\bx^* \in (\bx_0 + \mS_{G_1}) \cap \mathbb{R}^n_{>0}$.

\smallskip

For part (b), given $\bk \in \mK (G_1,G_1)$, there exists some $\bk' \in \mK (G_1)$, such that
\begin{equation}
\label{eq:key_4_3}
(G_1, \bk) \sim (G_1, \bk').
\end{equation}
Thus, by Theorem \ref{thm:cb}, the systems $(G_1, \bk)$ and $(G_1, \bk')$ share a unique steady state in $(\bx_0 + \mS_{G_1}) \cap \mathbb{R}^n_{>0}$, denoted by $\bx$. 

\smallskip

For part (c), define $\bJ' := (k'_{\by \to \by'} \bx^{\by})_{\by \to \by' \in E_1}$, then we construct a flux vector on $G_1$ as follows:
\begin{equation}
\label{eq:key_4_4}
\hbJ := \bJ' + \sum\limits^{a}_{i=1} (\langle \bJ^*, \bA_i \rangle - \langle \bJ', \bA_i \rangle) \bA_i.
\end{equation}
Under direct computation, we have
\begin{equation}
\label{eq:key_4_5}
\langle \hbJ, \bA_i \rangle = \langle \bJ^*, \bA_i \rangle
\ \text{ for any } \ 
1 \leq i \leq a.
\end{equation}
Note that $\bk' \in \mK (G_1)$ and $\{\bA_i \}^{a}_{i=1} \in \eJ(G) \subset \hat{\mathcal{J}} (G_1)$, then \eqref{eq:key_4_4} show that
\begin{equation}
\label{eq:key_4_5.5}
\bJ' \in \mathcal{J} (G_1)
\ \text{ and } \
\hbJ \in \hat{\mathcal{J}} (G_1).
\end{equation}
Consider the flux vector $\bJ := (k_{\by \to \by'} \bx^{\by})_{\by \to \by' \in E_1}$.
Using Proposition \ref{prop:craciun2020efficient} and \eqref{eq:key_4_3}, we deduce
\begin{equation} \notag
(G_1, \bJ) \sim (G_1, \bJ').
\end{equation}
From Lemma \ref{lem:j0}, this shows
$\bJ' - \bJ \in \mD (G_1)$.
Together with \eqref{eq:key_4_4}, we get
\begin{equation} \notag
\hbJ - \bJ \in \mD (G_1).
\end{equation}
Hence, we rewrite $\hbJ$ as
\begin{equation}
\label{eq:key_4_6}
\hbJ = \bJ + \bv  
\ \text{ with } \
\bv \in \mD (G_1).
\end{equation}
Now we set the reaction rate vector as
\begin{equation}
\label{eq:key_4_6.5}
\hbk := ( \frac{\hbJ}{\bx^{\by}} )_{\by \to \by' \in E_1} \in \RR^{E_1}.
\end{equation}

Using Proposition \ref{prop:craciun2020efficient} and \eqref{eq:key_4_6}, we obtain $(G_1, \bk) \sim (G_1, \hbk)$.
Together with \eqref{eq:key_4_5} and \eqref{eq:key_4_5.5}, we derive that the reaction rate vector $\hbk$ satisfies conditions \ref{lem:key_4_a}, \ref{lem:key_4_b} and \ref{lem:key_4_c}.

We now show the uniqueness of the vector $\hbk$. Suppose there exists another reaction rate vector $\hbk_1$ satisfying conditions \ref{lem:key_4_a}-\ref{lem:key_4_c}.
From the condition \ref{lem:key_4_a}, we have
\[
(G_1, \hbk) \sim (G_1, \hbk_1).
\]
From the condition \ref{lem:key_4_b}, we get 
\[
\hbJ_1 := (\hat{k}_{1, \by \to \by'} \bx^{\by})_{\by \to \by' \in E_1} \in \hat{\mathcal{J}} (G_1).
\]
Then Proposition \ref{prop:craciun2020efficient} and Lemma \ref{lem:j0} show
\[
(G_1, \hbJ) \sim (G_1, \hbJ_1)
\ \text{ and } \
\hbJ_1 - \hbJ \in \eJ (G_1).
\]
Using the condition \ref{lem:key_4_c}, we obtain
\[
\langle \hbJ, \bA_i \rangle = \langle \hbJ_1, \bA_i \rangle 
\ \text{ for any } \
1 \leq i \leq a.
\]
Since $\{\bA_i \}^{a}_{i=1}$ is an orthonormal basis of the subspace $\eJ(G)$, this implies that
\[
\hbJ_1 - \hbJ \in \big( \eJ (G_1) \big)^{\perp}.
\]
Hence, $\hbJ_1 - \hbJ = \mathbf{0}$ and $\hbk_1 = \hbk$. Therefore, we conclude the uniqueness.

\smallskip

For part (d), we will prove it in a sequence of three steps.

\smallskip

\textbf{Step 1: }
Assume a sequence of reaction rate vectors $\bk_i \in \mK (G_1,G_1)$ with $i \in \mathbb{N}$, such that 
\[
\bk_i \to \bk^*
\ \text{ as } \
i \to \infty.
\]
Analogously, there exists some $\bk'_i \in \mK (G_1)$, such that $(G_1, \bk_i) \sim (G_1, \bk'_i)$. 
Moreover, two systems $(G_1, \bk_i)$ and $(G_1, \bk'_i)$ share a unique steady state $\bx^i \in (\bx_0 + \mS_{G_1}) \cap \mathbb{R}^n_{>0}$.
Follow the steps in \eqref{eq:key_4_3}-\eqref{eq:key_4_5}, we obtain the corresponding sequences of flux vector as follows:
\begin{equation}
\begin{split} \label{eq:key_4_7}
& \bJ_i := (k_{i, \by \to \by'} (\bx^i)^{\by})_{\by \to \by' \in E_1}
\ \text{ with } \
i \in \mathbb{N},
\\& \bJ'_i := (k'_{i, \by \to \by'} (\bx^i)^{\by})_{\by \to \by' \in E_1}
\ \text{ with } \
i \in \mathbb{N}.
\end{split}
\end{equation}
and 
\begin{equation}
\label{eq:key_4_8}
\hbJ_i := \bJ'_i + \sum\limits^{a}_{j=1} (\langle \bJ^*, \bA_j \rangle - \langle \bJ'_i, \bA_j \rangle) \bA_j
\ \text{ with } \
i \in \mathbb{N}.
\end{equation}
Under direct computation, for any $i \in \mathbb{N}$,
\begin{equation}
\label{eq:key_4_8.5}
\langle \hbJ_i, \bA_j \rangle = \langle \bJ^*, \bA_j \rangle
\ \text{ for any } \ 
1 \leq j \leq a,
\end{equation}
and similar from \eqref{eq:key_4_5.5}, we have
\begin{equation}
\label{eq:key_4_12}
\hbJ_i \in \hat{\mathcal{J}} (G_1)
\ \text{ for any } \
i \in \mathbb{N}.
\end{equation}
Using Proposition \ref{prop:craciun2020efficient} and $(G_1, \bk_i) \sim (G_1, \bk'_i)$, we deduce 
\begin{equation} \notag
(G_1, \bJ_i) \sim (G_1, \bJ'_i)
\ \text{ for any } \
i \in \mathbb{N}.
\end{equation}
From Lemma \ref{lem:j0}, together with \eqref{eq:key_4_8}, we get
\begin{equation} \notag
\hbJ_i - \bJ_i \in \mD (G_1)
\ \text{ for any } \
i \in \mathbb{N}.
\end{equation}
Thus, for any $i \in \mathbb{N}$, $\hbJ_i$ can be expressed as
\begin{equation}
\label{eq:key_4_9}
\hbJ_i = \bJ_i + \bv^i
\ \text{ with } \
\bv^i \in \mD (G_1).
\end{equation}
On the other hand, using Lemma \ref{lem:key_2}, together with $\bk_i \to \bk^*$ as $i \to \infty$, we have
\begin{equation} \notag
\bx^i \to \bx^*
\ \text{ as } \ 
i \to \infty.
\end{equation}
Combining with \eqref{eq:key_4_7}, we derive that
\begin{equation}
\label{eq:key_4_10}
\bJ_i \to \bJ^*
\ \text{ as } \ 
i \to \infty.
\end{equation}

\smallskip

\textbf{Step 2: }
Now we claim that
\begin{equation} 
\label{eq:key_4_13}
\| \bv^i \|_{\infty} \to 0
\ \text{ as } \ 
i \to \infty.
\end{equation}
We prove this by contradiction. 
Suppose not, w.l.o.g. there exists a subsequence $\{\bv^{i_l} \}^{\infty}_{l=1}$, such that for any $l \in \mathbb{N}$,
\begin{equation} \notag
\| \bv^{i_l} \|_{\infty} \geq 1.
\end{equation}
Then we consider the sequence $\{ \bw^l \}^{\infty}_{l=1}$ as follows:
\begin{equation}
\label{eq:key_4_14}
\bw^{l} = \frac{\bv^{i_l}}{\| \bv^{i_l} \|_{\infty}}
\ \text{ with } \
l \in \mathbb{N}.
\end{equation}
It is clear that $\| \bw^{l} \|_{\infty} = 1$ for any $l \in \mathbb{N}$.
From the Bolzano–Weierstrass theorem, there exists a subsequence $\{ \bw^{l_j} \}^{\infty}_{j=1}$, such that
\begin{equation} \notag
\bw^{l_j} \to \bw^*
\ \text{ as } \ 
j \to \infty.
\end{equation}
Recall from \eqref{eq:key_4_9} and \eqref{eq:key_4_14}, we have for any $j \in \mathbb{N}$,
\begin{equation}
\label{eq:key_4_15}
\bw^{l_j} = \frac{\bv^{i_{l_j}}}{\| \bv^{i_{l_j}} \|_{\infty}} = \frac{1}{\| \bv^{i_{l_j}} \|_{\infty}}
\big( \hbJ_{i_{l_j}} - \bJ_{i_{l_j}} \big).
\end{equation}
Since $\bv^i \in \mD (G_1)$, together with $\| \bv^{i_l} \|_{\infty} \geq 1$, we obtain that 
\[
\bw^{l_j} \in \mD (G_1).
\]
Note that $\mD (G_1)$ is a linear subspace of finite dimension. Therefore, $\bw^{l_j} \to \bw^*$ implies
\begin{equation}
\label{eq:key_4_16}
\bw^* \in \mD (G_1).
\end{equation}

Let $\bz \in \big( \hat{\mathcal{J}} (G_1) \big)^{\perp}$.
From \eqref{eq:key_4_12}, we have for any $j \in \mathbb{N}$,
\begin{equation}
\label{eq:key_4_17}
\langle \hbJ_{i_{l_j}}, \bz \rangle = 0.
\end{equation}
From \eqref{eq:key_4_10} and $\bJ \in \mathcal{J} (G_1)$, we obtain 
\begin{equation}
\label{eq:key_4_18}
\langle \bJ_{i_{l_j}}, \bz \rangle \to \langle \bJ,  \bz \rangle = 0
\ \text{ as } \ 
j \to \infty.
\end{equation}
Using \eqref{eq:key_4_15}, \eqref{eq:key_4_17} and \eqref{eq:key_4_18}, together with $\| \bv^{i_l} \|_{\infty} \geq 1$ and $\bw^{l_j} \to \bw^*$, we derive 
\begin{equation} \notag
\langle \bw^{l_j}, \bz \rangle \to \langle \bw^*, \bz \rangle = 0.
\end{equation}
Since $\bz$ is arbitrary in $\big( \hat{\mathcal{J}} (G_1) \big)^{\perp}$, this shows $\bw^* \in \hat{\mathcal{J}} (G_1)$. 
Together with \eqref{eq:key_4_16}, we get
\begin{equation}
\label{eq:key_4_19}
\bw^* \in \eJ (G_1).
\end{equation}

Recall that $\{\bA_i \}^{a}_{i=1}$ is an orthonormal basis of the subspace $\eJ(G)$. Without loss of generality, we pick $\bA_1 \in \eJ(G)$. From \eqref{eq:key_4_8.5} and \eqref{eq:key_4_10}, we get
\begin{equation} \notag
\langle \hbJ_{i_{l_j}} - \bJ_{i_{l_j}}, \bA_1 \rangle
= \langle \bJ^*, \bA_1 \rangle - \langle \bJ_{i_{l_j}}, \bA_1 \rangle \to 0
\ \text{ as } \ 
j \to \infty.
\end{equation}
Together with $\| \bv^{i_l} \|_{\infty} \geq 1$ and $\bw^{l_j} \to \bw^*$, we derive 
\begin{equation} \notag
\langle \bw^{l_j}, \bA_1 \rangle \to \langle \bw^*,  \bA_1 \rangle = 0.
\end{equation}
Analogously, we can get
$\langle \bw^*,  \bA_j \rangle = 0$ for any $1 \leq j \leq a$.
This shows that
\begin{equation}
\label{eq:key_4_20}
\bw^* \in \big( \eJ (G_1) \big)^{\perp}.
\end{equation}
Combining \eqref{eq:key_4_19} with \eqref{eq:key_4_20}, we conclude that $\bw^* = \mathbf{0}$. Since $\| \bw^{l} \|_{\infty} = 1$ for any $l \in \mathbb{N}$, this contradicts with $\bw^{l_j} \to \bw^*$ as $j \to \infty$.
Therefore, we prove the claim.

\smallskip

\textbf{Step 3: }
Using \eqref{eq:key_4_9}, \eqref{eq:key_4_10} and \eqref{eq:key_4_13}, we derive that
\begin{equation}
\label{eq:key_4_21}
\hbJ_i = \bJ_i + \bv^i \to \bJ^* 
\ \text{ as } \
i \to \infty.
\end{equation}
Since $\bJ \in \mathcal{J} (G_1) \subset \RR^{E_1}_{>0}$, there exists sufficiently large $N$, such that
\begin{equation} \notag
\hbJ_i \in \RR^{E_1}_{>0}
\ \text{ for any } \
i > N.
\end{equation}
Together with \eqref{eq:key_4_12} and Remark \ref{rmk:hat_j_g1_g}, we obtain that 
\[
\hbJ_i \in \hat{\mathcal{J}} (G_1) \cap \RR^{|E_1|}_{>0} = \mathcal{J} (G_1)
\ \text{ for any } \
i > N.
\]
Following \eqref{eq:key_4_6.5}, we set $\{ \hbk_i\}^{\infty}_{i=1}$ as follows: 
\begin{equation}
\label{eq:key_4_22}
\hbk_i := \big( \frac{\hat{J}_{i, \by \to \by'} }{(\bx^i)^{\by}} \big)_{\by \to \by' \in E_1}
\ \text{ with } \
i \in \mathbb{N}.
\end{equation}
Note that $\bx^i \in (\bx_0 + \mS_{G_1}) \cap \mathbb{R}^n_{>0}$ and $\hbJ_i \in \mathcal{J} (G_1)$ for any $i > N$, we get 
\begin{equation} \notag
\hbk_i \in \mK (G_1)
\ \text{ for any } \
i > N.
\end{equation}
Using \eqref{eq:key_4_9} and Proposition \ref{prop:craciun2020efficient}, we derive
\begin{equation} \notag
(G_1, \bk_i) \sim (G_1, \hbk_i).
\end{equation}
Finally, using $\hbJ_i \to \bJ^*$ and $\bx^i \to \bx^*$, together with $\bJ^* = (k^*_{\by \to \by'} (\bx^*)^{\by})_{\by \to \by' \in E_1}$, we have
\begin{equation}
\hbk_i \to \bk^*
\ \text{ as } \
i \to \infty.
\end{equation}
Therefore, we conclude the proof of this Proposition.
\end{proof}

Now we are ready to prove Proposition~\ref{prop:inverse_cts_k}.

\begin{proof}[Proof of Proposition \ref{prop:inverse_cts_k}]

Given fixed $\bq = (q_1, q_2, \ldots, q_a) \in \RR^a$, consider $\bk \in \dK(G,G_1)$ such that
\begin{equation} \notag
\Phi (\bk, \bq) = (\hat{\bJ},\bx, \bp).
\end{equation}
Follow definition, there exists $\bk_1 \in \mK (G_1) \subset \mK_{\RR} (G_1,G)$ satisfying 
\[
(G, \bk) \sim (G_1, \bk_1).
\]
Remark \ref{rmk:de_ss} shows $\bx \in (\bx_0 + \mS_{G_1} )\cap\mathbb{R}^n_{>0}$ is the steady state of $(G_1, \bk_1)$ and $(G, \bk)$.
From Lemma \ref{lem:phi_wd}, by setting
\begin{equation}
\label{eq:cts_k_1}
\bJ = \big( k_{1, \by\rightarrow \by'} \bx^{\by} \big)_{\by\rightarrow \by' \in E_1},
\end{equation}
then we obtain
\begin{equation}
\label{eq:cts_k_2}
\hbJ = \bJ + \sum\limits^a_{j=1} (q_j - \langle \bJ, \bA_j \rangle ) \bA_j \in \hat{\mJ} (G_1,G).
\end{equation}
Moreover, from \eqref{def:phi_kq} we obtain 
\begin{equation} \notag
\bp = ( \langle \bk, \bB_1 \rangle, \langle \bk, \bB_2 \rangle, \ldots, \langle \bk, \bB_b \rangle),
\end{equation}
which is continuous with respect to $\bk$.

\smallskip

Now assume any sequence $\{ \bk^i \}^{\infty}_{i = 1}$ in $\dK(G,G_1)$, such that
\begin{equation} 
\label{eq:cts_k_3}
\bk^i \to \bk
\ \text{ as } \
i \to \infty.
\end{equation}
Suppose $\Phi (\bk^i, \bq) = (\hbJ^i, \bx^i, \bp^i)$ with $i \in \mathbb{N}$, then $\bx^i \in (\bx_0 + \mS_{G_1} )\cap\mathbb{R}^n_{>0}$ is the steady state of $(G_1, \bk^i)$.
Using Lemma \ref{lem:key_3}, together with $\bk^i \to \bk$ in \eqref{eq:cts_k_3}, we derive
\begin{equation}
\label{eq:cts_k_4}
\bx^i \to \bx
\ \text{ as } \
i \to \infty.
\end{equation}
From Lemma \ref{lem:key_1}, there exists a sequence $\{ \bk^i_1 \}^{\infty}_{i = 1}$ in $\mK (G_1,G_1)$, such that
\begin{equation} \notag
(G, \bk^i) \sim (G_1, \bk^i_1)
\ \text{ for any } \
i \in \mathbb{N},
\end{equation}
and
\begin{equation}
\label{eq:cts_k_5}
\bk^i_1 \to \bk_1
\ \text{ as } \
i \to \infty.
\end{equation}
Then apply Proposition \ref{lem:key_4}, there exists a corresponding sequence $\{ \hbk_i \}^{\infty}_{i = 1}$, such that
\begin{equation} \notag
(G_1, \hbk_i) \sim (G_1, \bk^i_1)
\ \text{ for any } \
i \in \mathbb{N},
\end{equation}
Set $\hbJ_i = (\hat{k}_{i, \by \to \by'} (\bx^i)^{\by})_{\by \to \by' \in E_1}$, then for any $i \in \mathbb{N}$,
\begin{equation}
\label{eq:cts_k_6}
\hbJ_i \in \hat{\mathcal{J}} (G_1)
\ \text{ and } \
\langle \hbJ_i, \bA_j \rangle = \langle \bJ, \bA_j \rangle
\ \text{ for any } \
1 \leq j \leq a.
\end{equation}
Moreover, from $\bk^i_1 \to \bk_1$ in \eqref{eq:cts_k_5}, we have
\begin{equation} \notag
\hbk_i \to \bk_1 
\ \text{ as } \
i \to \infty.
\end{equation} 
Together with $\bx^i \to \bx$ in \eqref{eq:cts_k_4} and $\bJ$ in \eqref{eq:cts_k_1}, we derive that 
\begin{equation}
\label{eq:cts_k_7}
\hbJ_i \to \bJ 
\ \text{ as } \
i \to \infty.
\end{equation}
Since $\bJ \in \mathcal{J} (G_1)$ and $\hbJ_i \in \hat{\mathcal{J}} (G_1)$, this shows there exists a sufficiently large $N$, such that 
\begin{equation}
\label{eq:cts_k_8}
\hbJ_i \in \mathcal{J} (G_1)
\ \text{ for any } \
i > N.
\end{equation}
Note that $(G_1, \hbk_i) \sim (G_1, \bk^i_1) \sim (G_1, \bk^i)$, thus $\bx^i$ is also the steady state of $(G_1, \hbk_i)$.
Since $\hbJ_i = (\hat{k}_{i, \by \to \by'} (\bx^i)^{\by})_{\by \to \by' \in E_1}$, together with \eqref{eq:cts_k_8}, we deduce
\begin{equation} \notag
\hbk_i \in \mK (G_1)
\ \text{ for any } \
i > N.
\end{equation}
Note that $\Phi (\bk^i, \bq) = (\hbJ^i, \bx^i, \bp^i)$. 
From \eqref{eq:cts_k_2},  we obtain
\begin{equation} \notag
\hbJ^i = \hbJ_i + \sum\limits^a_{j=1} (q_j - \langle \hbJ_i, \bA_j \rangle ) \bA_j
\ \text{ for any } \
i > N.
\end{equation}
Using \eqref{eq:cts_k_6} and \eqref{eq:cts_k_7}, we have
\begin{equation} \notag
\hbJ^i \to \bJ
\ \text{ as } \
i \to \infty.
\end{equation}

Recall that $\Phi (\bk, \bq) = (\bJ, \bx, \bp)$. Suppose any sequence $\bk^i \to \bk$ with $\Phi (\bk^i, \bq) = (\hbJ^i, \bx^i, \bp^i)$, we show the continuity on $\bp$, $\bx^i \to \bx$ and $\hbJ^i \to \bJ$. Therefore, we conclude that $\Phi (\cdot, \bq)$ is continuous with respect to $\bk$.
\end{proof}

Here we state the first main theorem in this paper.

\begin{theorem}
\label{thm:inverse_cts}
Consider the map $\hPsi$ in Definition \ref{def:hpsi}, then the map $\hPsi^{-1}$ is continuous.
\end{theorem}

\begin{proof}

From Lemma \ref{lem:phi_wd}, consider the map $\Phi$ in Definition \ref{def:phi}, then $\Phi = \hPsi^{-1}$ is well-defined and bijective. Thus, it suffices to show the map $\Phi$ is continuous.

Suppose any $(\bk, \bq) \in \dK(G,G_1) \times \RR^a$. Consider any positive real number $\varepsilon > 0$. 
From Proposition \ref{prop:inverse_cts_k}, $\Phi (\cdot, \bq)$ is continuous with respect to $\bk$. Thus, there exists some positive real number $\delta_1 > 0$, such that for any $\tilde{\bk} \in \dK(G,G_1)$ with $\| \tilde{\bk} - \bk \| < \delta_1$, then
\begin{equation}
\label{eq:inverse_cts_1}
\big\| \Phi (\tilde{\bk}, \bq) - \Phi (\bk, \bq) \big\| < \frac{\varepsilon}{2}.
\end{equation}
Note that $\{\bA_1, \bA_2, \ldots, \bA_a \}$ is an orthonormal basis of $\eJ(G_1) \subset \RR^a$, there exists some positive real number $\delta_2 > 0$, such that for any $\bv = (v_1, v_2, \ldots, v_a) \in \RR^a$ with $\| \bv \| < \delta_2$, then
\begin{equation}
\label{eq:inverse_cts_2}
\big\| \sum\limits^{a}_{i=1} v_i \bA_i \big\| < \frac{\varepsilon}{2}.
\end{equation}

Let $\delta = \min \{ \delta_1, \delta_2 \}$, consider any $(\hbk, \hbq) \in \dK(G,G_1) \times \RR^a$ with $| (\hbk, \hbq) - (\bk, \bq) | < \delta$.
This implies $\| \hbk - \bk \| < \delta$ and $\| \hbq - \bq \| < \delta$. Then we compute that
\begin{equation}
\label{eq:inverse_cts_3}
\Phi (\hbk, \hbq) - \Phi (\bk, \bq)
= \big( \Phi (\hbk, \hbq) - \Phi (\bk, \hbq) \big) + \big( \Phi (\bk, \hbq) - \Phi (\bk, \bq) \big).
\end{equation}
From \eqref{eq:inverse_cts_1} and $\| \hbk - \bk \| < \delta \leq \delta_1$, we have
\begin{equation}
\label{eq:inverse_cts_4}
\big\| \Phi (\hbk, \hbq) - \Phi (\bk, \hbq) \big\| < \frac{\varepsilon}{2}.
\end{equation}
Using Lemma \ref{lem:inverse_cts_q} and setting $\hbq - \bq := (v_1, v_2, \ldots, v_a) \in \RR^a$, we have
\begin{equation} \notag
\Phi (\bk, \hbq) - \Phi (\bk, \bq) = \sum\limits^{a}_{i=1} v_i \bA_i,
\end{equation}
Together with \eqref{eq:inverse_cts_2} and $\| \hbq - \bq \| < \delta \leq \delta_2$, we obtain
\begin{equation}
\label{eq:inverse_cts_5}
\big\| \Phi (\bk, \hbq) - \Phi (\bk, \bq) \big\|
= \big\| \sum\limits^{a}_{i=1} v_i \bA_i \big \| < \frac{\varepsilon}{2}.
\end{equation}
Inputting \eqref{eq:inverse_cts_4} and \eqref{eq:inverse_cts_5} into \eqref{eq:inverse_cts_3}, we derive
\begin{equation} \notag
\big\| \Phi (\hbk, \hbq) - \Phi (\bk, \bq) \big\|
\leq \frac{\varepsilon}{2} + \frac{\varepsilon}{2} = \varepsilon.
\end{equation}
Therefore, $\Phi$ is continuous and we conclude this theorem.
\end{proof}

The following result is a direct consequence of Theorem \ref{thm:inverse_cts}.

\begin{theorem}
\label{thm:hpsi_homeo}
The map $\hPsi$ in Definition \ref{def:hpsi}
is a homeomorphism.
\end{theorem}

\begin{proof}

From Lemma \ref{lem:hpsi_bijective} and \ref{lem:hpsi_cts}, we derive that $\hPsi$ is bijective and continuous.

On the other hand, Proposition \ref{thm:inverse_cts} shows the inverse map $\hPsi^{-1}$ is also continuous.
Therefore, we conclude that the map $\hPsi$ is a homomorphism.
\end{proof}

\section{Dimension of \texorpdfstring{$\dK(G,G_1)$}{KGG1} and \texorpdfstring{$\pK(G,G_1)$}{pKGG1} }
\label{sec:dimension}

In this section, we give a precise bound on the dimension of $\dK(G, G_1)$, where $G_1 \sqsubseteq G_c$. Further, we show the dimension of $\pK(G, G_1)$ when $\pK(G, G_1) \neq \emptyset$. Finally, we remark on the dimension of {\em $\RR$-disguised toric locus} $\dK(G)$ and {\em disguised toric locus} $\pK(G)$.

\begin{lemma}
\label{lem:hat_j_g1_g_cone}
Let $G_1 = (V_1, E_1)$ be a weakly reversible E-graph and let $G = (V, E)$ be an E-graph. 
If $\mJ (G_1, G) \neq \emptyset$, then $\hat{\mJ} (G_1, G)$ is a convex cone, which satisfies
\begin{equation} \label{hat_j_g1_g_generator_dim}
\dim (\hat{\mJ} (G_1, G)) = \dim (\mJ (G_1, G)).
\end{equation}
\end{lemma}

\begin{proof}

From Lemma \ref{lem:j_g1_g_cone}, suppose there exists a set of vectors $\{ \bv_1, \bv_2, \ldots, \bv_k \} \subset \RR^{|E_1|}$, such that 
\begin{equation} \notag
\mJ (G_1, G) = \{ a_1 \bv_1 + \cdots a_k \bv_k \ | \ a_i \in \RR_{>0} \}.
\end{equation}
Using \eqref{def:hat_j_g1_g}, $\hat{\mJ} (G_1, G)$ can be represented as the positive combination of the following vectors: 
\begin{equation} \label{hj_g1g_basis}
\{ \bv_1, \bv_2, \ldots, \bv_k, \pm \bA_1, \pm \bA_2, \ldots, \pm \bA_a \}. 
\end{equation}
This shows $\hat{\mJ} (G_1, G)$ is a convex cone. Moreover, we have
\begin{equation} \notag
\dim (\hat{\mJ} (G_1, G)) =\dim ( \spn \{ \bv_1, \bv_2, \ldots, \bv_k, \bA_1, \bA_2, \ldots, \bA_a \} ).
\end{equation}
Since $\mJ (G_1, G) \neq \emptyset$, Lemma \ref{lem:j_g1_g_cone} shows that
\begin{equation} \notag
\spn \{ \bA_i \}^a_{i=1} = \eJ(G_1) \subseteq \spn \{ \bv_1, \bv_2, \ldots, \bv_k \}.
\end{equation}
Therefore, we conclude that
\begin{equation} \notag
\dim (\hat{\mJ} (G_1, G)) = \dim ( \spn \{ \bv_1, \bv_2, \ldots, \bv_k \} ) = \dim (\mJ (G_1, G)).
\end{equation} 
\end{proof}

\begin{theorem}
\label{thm:dim_kisg}

Let $G_1 = (V_1, E_1)$ be a weakly reversible E-graph with its stoichiometric subspace $\mS_{G_1}$. Suppose an E-graph $G = (V, E)$, recall $\mJ (G_1,G)$, $\mD(G)$ and $\eJ(G_1)$ defined in Definitions~\ref{def:flux_realizable}, \ref{def:d0} and \ref{def:j0} respectively.

\begin{enumerate}[label=(\alph*)]
\item\label{part_a} Consider $\dK(G,G_1)$ from Definition~\ref{def:de_realizable}, then
\begin{equation} \label{eq:dim_kisg}
\begin{split} 
& \dim(\dK(G,G_1)) 
= \dim (\mJ(G_1,G)) + \dim (\mS_{G_1})  + \dim(\eJ(G_1)) - \dim(\mD(G)).
\end{split}
\end{equation}

\item\label{part_b} Further, consider $\pK (G, G_1)$ from Definition~\ref{def:de_realizable} and assume that $\pK (G, G_1) \neq \emptyset$. Then
\begin{equation} \label{eq:dim_kdisg}
\dim(\pK (G,G_1)) = \dim(\dK(G,G_1)).
\end{equation}

\end{enumerate}

\end{theorem}

\begin{proof}

For part $(a)$, recall we prove that $\hat{\Psi}$ is a homeomorphism in Theorem \ref{thm:hpsi_homeo}. Using the invariance of dimension theorem \cite{hatcher2005algebraic,munkres2018elements}, together with  Remark \ref{rmk:semi_algebaic} and
\eqref{hat_j_g1_g_generator_dim} in Lemma \ref{lem:hat_j_g1_g_cone}, we obtain
\begin{equation} \notag
\dim (\dK(G, G_1)) + \dim(\mD(G)) =
\dim (\mJ (G_1, G)) + \dim (\mS_{G_1}) + \dim(\eJ(G_1)),
\end{equation}
and conclude \eqref{eq:dim_kisg}. 
Further, we emphasize that on a dense open subset of $\dK(G, G_1)$, it is locally a submanifold.
The homomorphism indicates that all such submanifolds have the same dimension.

\smallskip

For part $(b)$, since $\pK (G, G_1) \neq \emptyset$, together with Lemma \ref{lem:semi_algebaic} and Remark \ref{rmk:semi_algebaic}, there exists a $\bk \in \pK(G, G_1)$ and a neighborhood of $\bk$ in $\pK(G, G_1)$, denoted by $U$, such that 
\[
\bk \in U \subset \pK(G, G_1),
\]
where $U$ is a submanifold with $\dim (U) = \dim (\pK(G, G_1))$. Moreover, $\pK (G, G_1) = \dK(G, G_1) \cap \mathbb{R}^{E}_{>0}$ implies that $U$ is also a neighborhood of $\bk$ in $\dK(G, G_1)$. 
From part $(a)$, we obtain that on a dense open subset of $\dK(G, G_1)$, all local submanifolds have the same dimension. Therefore, we conclude \eqref{eq:dim_kdisg}.

\end{proof}

\begin{theorem} \label{thm:dim_kisg_main}

Consider an E-graph $G = (V, E)$. 

\begin{enumerate}[label=(\alph*)]
\item Consider $\dK(G)$ from Definition~\ref{def:de_realizable}, then
\begin{equation} \notag
\dim (\dK(G) )
= \max_{G'\sqsubseteq G_c} 
\Big\{ \dim (\mJ(G',G)) + \dim (\mS_{G'})  + \dim(\eJ(G')) - \dim(\mD(G)) 
\Big\},
\end{equation}
where $\mJ (G',G)$, $\mD(G)$ and $\eJ(G')$ are defined in Definitions~\ref{def:flux_realizable}, \ref{def:d0} and \ref{def:j0} respectively.

\item Further, consider $\pK (G)$ from Definition~\ref{def:de_realizable} and assume that $\pK (G) \neq \emptyset$. Then
\begin{equation} \notag
\begin{split}
& \dim (\pK(G) )
\\& = \max_{ \substack{ G'\sqsubseteq G_c, \\  \pK(G, G') \neq \emptyset } } 
\Big\{ \dim (\mJ(G',G)) + \dim (\mS_{G'})  + \dim(\eJ(G')) - \dim(\mD(G)) 
\Big\}.
\end{split}
\end{equation}
\end{enumerate}
\end{theorem}

\begin{proof}

Note that from Lemma~\ref{lem:semi_algebaic}, both $\pK (G, G_1)$ and $\dK(G, G_1)$ are semialgebraic sets. 
 
Further, the dimension of the union of finitely many semialgebraic sets is the maximum of the dimensions of these semialgebraic sets \cite{coste2000introduction,basu201738,lairez2021computing}. 
The result then follows from Definition~\ref{def:de_realizable}, Lemma \ref{lem:semi_algebaic} and Theorem~\ref{thm:dim_kisg}.
\end{proof}

\section{Examples} \label{sec:applications}

\begin{example}[Thomas type models {\cite[Chapter 6]{murray2007mathematical}}]
\label{ex:thomas}

This model illustrates the substrate inhibition mechanism, capturing the interaction between uric acid and oxygen catalyzed by the enzyme uricase. The E-graph corresponding to this model, shown in Figure~\ref{fig:thomas_model} (a), is denoted by $G$. The system is governed by the following differential equations:
\begin{equation}
\begin{split}
\frac{du}{dt} & = c - u - uv,
\\ \frac{dv}{dt} & = \beta(d-v) - uv, 
\end{split}
\end{equation}
where $u$ denotes the concentration of uric acid, $v$ represents the concentration of oxygen, and the parameters $c$, $d$, and $\beta$ are positive constants. 
As demonstrated in \cite{CraciunDickensteinShiuSturmfels2009}, the toric locus associated with $G$ forms a toric variety of codimension one, implying that it has measure zero.

\begin{figure}[!ht]
\centering
\includegraphics[scale=0.7]{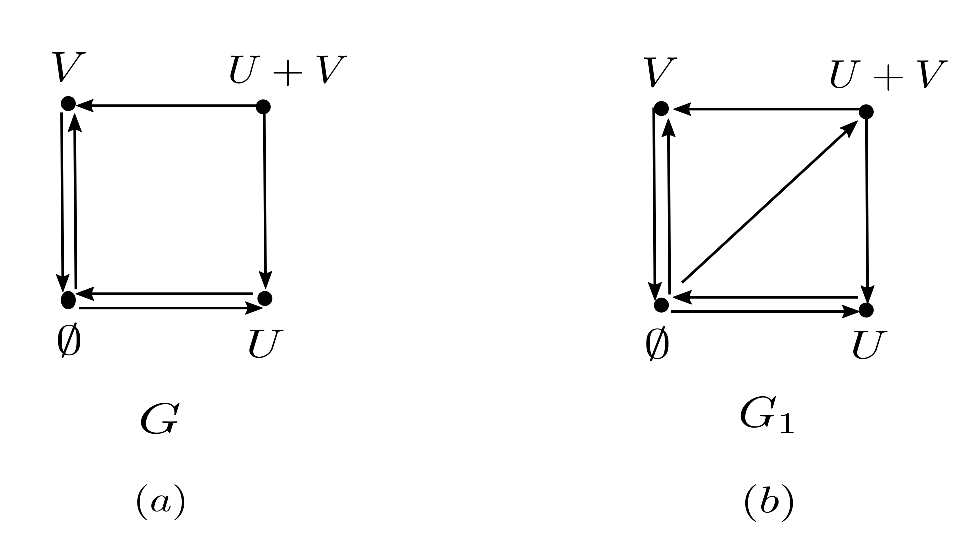}
\caption{ 
(a) The E-graph $G$ represents a Thomas-type model, where $U$ denotes uric acid and $V$ denotes oxygen.
(b) The E-graph $G_1$ is a weakly reversible subgraph of the complete graph formed by the source vertices of $G$.
}
\label{fig:thomas_model}
\end{figure}

The graph $G$ contains 6 reactions, hence $\pK(G) \subseteq \mathbb{R}^6_{>0}$. We claim that 
\[
\rm{dim}(\pK(G)) = 6.
\]
This implies that the disguised toric locus corresponding to $G$ is of positive measure. To prove this claim, we consider the E-graph $G_1$ and compute $\dim (\pK(G, G_1))$.

First, consider a flux vector $\bJ \in \mJ(G_1, G) \subset \mathbb{R}^7$. Since $\bJ$ is a complex-balanced flux vector in $G_1$, this imposes $4 - 1 = 3$ constraints on $\bJ$, as shown in \cite{craciun2020efficient}. Additionally, being $\mathbb{R}$-realizable in $G$ imposes no further constraints, since every flux vector in $G_1$ can be transformed into $G$. Therefore, we have
\[
\dim (\mJ(G_1, G)) = 7 - 0 - 3 = 4.
\]
Second, upon inspecting the graphs $G$ and $G_1$, we obtain the following:
\[
\dim (\mathcal{S}_{G_1}) = 2, \ \
\rm{dim}(\mD (G))  = 0
\ \text{ and } \
\dim (\eJ(G_1)) =  0.
\]
Using Theorem~\ref{thm:dim_kisg} (a), we get that
\begin{equation}
\begin{split} \notag
\rm{dim}(\dK(G, G_1)) 
& = \dim (\mJ(G_1, G)) + \dim (\mathcal{S}_{G_1}) + \dim(\mD (G)) - \dim(\eJ (G_1)) 
\\& = 4 + 2 + 0 - 0 = 6.
\end{split}
\end{equation}

Since $G_1$ is weakly reversible, $\mK ({G_1}) \neq \emptyset$. From Figure~\ref{fig:thomas_model}, given any $\bk_1 \in \mK ({G_1})$ there exists $\bk$ such that $(G,\bk)\sim (G_1, \bk_1)$. This implies that $\pK (G, G_1) \neq \emptyset$, and thus  Theorem~\ref{thm:dim_kisg} (b) further shows
\[
\rm{dim}(\pK(G,G_1))
= \rm{dim}(\dK(G,G_1)) = 6.
\]
This, together with Theorem \ref{thm:dim_kisg_main}, implies that $\rm{dim}(\pK(G)) = 6$.

A lengthy computation (based on the matrix-tree theorem, see \cite{CraciunDickensteinShiuSturmfels2009}) can be used to derive that  
\[
\pK(G) \ \supseteq \ \{ \bk \in \mathbb{R}^{7}_{>0} \mid \frac{k_{V \to \emptyset} k_{U \to \emptyset}}{k_{U+V \to U}} > k_{\emptyset \to V} \geq k_{\emptyset \to U} 
\ \text{ or } \
\frac{k_{V \to \emptyset} k_{U \to \emptyset}}{k_{U+V \to U}} > k_{\emptyset \to U} \geq k_{\emptyset \to V}
\},
\]
which, informally speaking, implies that at least 50\% of all positive parameter choices belong to $\pK(G)$.

Finally, we remark that the Thomas-type model $G$ and the corresponding weakly reversible E-graph $G_1$ in Figure~\ref{fig:thomas_model} are both two-dimensional.
From Remark \ref{rmk:complex_balance_property}, for any $\bk \in \pK(G, G_1)$, the system $(G, \bk)$ has a globally attracting steady state within each stoichiometric compatibility class.
\qed
\end{example}

\begin{example}[Circadian clock models \cite{leloup1999chaos}]
\label{ex:circadian}

The circadian clock forms an integral part of the behavioral, physical, and biological changes in the body over a 24-hour cycle. 
The dynamics of this network are known to exhibit exotic behaviors, such as oscillations and limit cycles. 
For our analysis, we utilize the model of the circadian clock introduced in \cite{leloup1999chaos}, described by the following reactions:
\begin{equation} \notag
P + T \rightleftharpoons C \rightarrow \emptyset, \ \
P \rightleftharpoons \emptyset, \ \ 
T \rightleftharpoons \emptyset,
\end{equation}
where $P$ denotes period, $T$ denotes time, and $C$ denotes the period-time complex. The E-graph corresponding to this model, shown in Figure~\ref{fig:circadian_clock} (a), is denoted by $G$. 
Since $G$ is not weakly reversible, \cite{CraciunDickensteinShiuSturmfels2009} demonstrates that the toric locus associated with $G$ is empty.

\begin{figure}[!ht]
\centering
\includegraphics[scale=0.43]{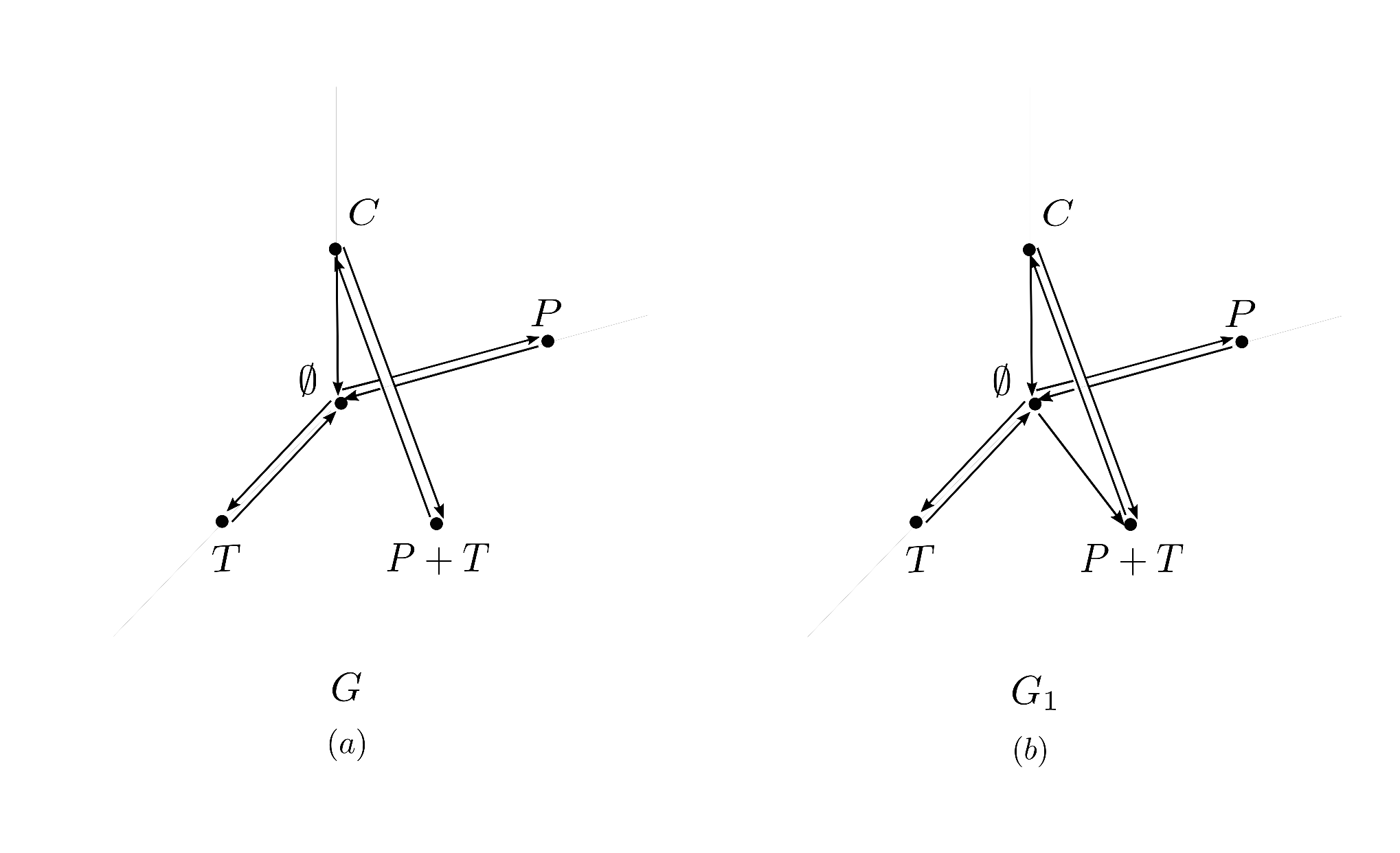}
\caption{
(a) The E-graph $G$ represents a circadian clock model.
(b) The E-graph $G_1$ is a weakly reversible subgraph of the complete graph formed by the source vertices of $G$.
}
\label{fig:circadian_clock}
\end{figure}

Note that the graph $G$ contains 7 reactions, so $\pK(G) \subseteq \mathbb{R}^7_{>0}$. We claim that 
\[
\rm{dim}(\pK(G)) = 7.
\]
This implies that the disguised toric locus corresponding to $G$ is of positive measure. To prove this claim, we consider the E-graph $G_1$ and compute $\dim (\pK(G, G_1))$.

First, consider a flux vector $\bJ \in \mJ(G_1, G) \subset \mathbb{R}^8$. Since $\bJ$ is a complex-balanced flux vector in $G_1$, this imposes $5 - 1 = 4$ constraints on $\bJ$, as shown in \cite{craciun2020efficient}. Additionally, being $\mathbb{R}$-realizable in $G$ imposes no further constraints, since every flux vector in $G_1$ can be transformed into $G$. Therefore, we have
\[
\dim (\mJ(G_1, G)) = 8 - 0 - 4 = 4.
\]
Second, upon inspecting the graphs $G$ and $G_1$, we obtain the following:
\[
\dim (\mathcal{S}_{G_1}) = 3, \ \
\rm{dim}(\mD (G))  = 0
\ \text{ and } \
\dim (\eJ(G_1)) =  0.
\]
Using Theorem~\ref{thm:dim_kisg} (a), we get that
\begin{equation}
\begin{split} \notag
\rm{dim}(\dK(G, G_1)) 
& = \dim (\mJ(G_1, G)) + \dim (\mathcal{S}_{G_1}) + \dim(\mD (G)) - \dim(\eJ (G_1)) 
\\& = 4 + 3 + 0 - 0 = 7.
\end{split}
\end{equation}

Since $G_1$ is weakly reversible, $\mK ({G_1}) \neq \emptyset$. From Figure~\ref{fig:circadian_clock}, given any $\bk_1 \in \mK ({G_1})$ there exists $\bk$ such that $(G,\bk)\sim (G_1, \bk_1)$. This implies that $\pK (G, G_1) \neq \emptyset$, and thus  Theorem~\ref{thm:dim_kisg}~(b) further shows
\[
\rm{dim}(\pK(G,G_1) 
= \rm{dim}(\dK(G,G_1)) = 7.
\]
This, together with Theorem \ref{thm:dim_kisg_main}, implies that $\rm{dim}(\pK(G)) = 7$.
Actually, a  computation (based on the matrix-tree theorem~\cite{CraciunDickensteinShiuSturmfels2009}) can be used  to derive that 
 $\pK(G) = \mathbb{R}^{7}_{>0}.
$

Finally, note that for the circadian clock model $G$, the corresponding weakly reversible E-graph $G_1$ in Figure~\ref{fig:circadian_clock} contains a single linkage class.
From Remark \ref{rmk:complex_balance_property} it follows that for any $\bk \in \pK(G, G_1)$, the system $(G, \bk)$ has a globally attracting steady state within each stoichiometric compatibility class.
\qed
\end{example}

\section{Discussion} \label{sec:discussion}

Due to their remarkably robust dynamics, complex-balanced dynamical systems form an important class of dynamical systems. In particular, they are notable for their connection to the \emph{Global Attractor Conjecture}~\cite{horn1972general, CraciunDickensteinShiuSturmfels2009}, which asserts that these systems possess a globally attracting steady state within each invariant polyhedron. The set of rate constants of a network that generate complex-balanced systems form a variety called the \emph{toric locus}~\cite{craciun2020structure}, and its codimension is determined by the deficiency of the reaction network~\cite{CraciunDickensteinShiuSturmfels2009}.

A generalization of the toric locus concerns the set of rate constants that yield the same dynamics as complex-balanced systems. When the rate constants can take both positive and negative values, this set is referred to as the \emph{$\mathbb{R}$-disguised toric locus}; if only positive values are allowed, it is called the \emph{disguised toric locus} \cite{2022disguised}.
In~\cite{disg_1}, it was demonstrated that both the disguised toric locus and the $\mathbb{R}$-disguised toric locus are path-connected. A recent study~\cite{disg_2} established a lower bound for the dimensions of both loci.

The key contribution of this paper is to estimate the exact dimensions of both the disguised and $\mathbb{R}$-disguised toric loci through the construction of an explicit homomorphism. Specifically, Theorem~\ref{thm:dim_kisg_main} provides these dimensions as follows:
{\small
\begin{equation} \notag
\begin{split}
&\dim (\dK(G) )
= \max_{G'\sqsubseteq G_c} 
\Big\{ \dim (\mJ(G',G)) + \dim (\mS_{G'})  + \dim(\eJ(G')) - \dim(\mD(G)) 
\Big\},
\\& \dim (\pK(G) )
= \max_{ \substack{ G'\sqsubseteq G_c, \\  \pK(G, G') \neq \emptyset } } 
\Big\{ \dim (\mJ(G',G)) + \dim (\mS_{G'})  + \dim(\eJ(G')) - \dim(\mD(G)) 
\Big\}.
\end{split}
\end{equation}
}
This characterization allows us to identify regions of parameter space that produce robust dynamics, thus refining the modeling and prediction of system responses.

The results outlined above generate several new directions for future research. One direction involves the efficient computation of the dimensions of both the $\mathbb{R}$-disguised toric locus and the disguised toric locus for a given network. Among the various quantities on the right-hand side of the equations, the most challenging to characterize is $\dim (\mJ(G', G))$, which we specifically aim to address in our upcoming work~\cite{disg_4}. Specifically, computing $\dim (\mJ(G', G))$ is equivalent to solving a linear feasibility problem that incorporates both dynamical equivalence and complex-balancing constraints.

Another promising direction is to identify the conditions under which the homeomorphisms constructed in this paper can be shown to be diffeomorphisms. A recent work~\cite{smoothness} has shown that the toric locus of a weakly reversible network is a smooth manifold.
Consequently, we aim to establish the conditions on the network that ensure the disguised toric locus is also a smooth manifold, with the corresponding steady states varying continuously along this locus.

\section*{Acknowledgements}

This work was supported in part by the National Science Foundation grant DMS-2051568.

\bibliographystyle{unsrt}
\bibliography{Bibliography}

\end{document}